%% file: paper_ver7-6.tex
\documentclass[12pt, draftclsnofoot, onecolumn]{IEEEtran}

\usepackage{amsmath, mathrsfs}
\usepackage{amsfonts}
\usepackage{amssymb}
\usepackage{graphicx}
\usepackage{color}
\usepackage{multirow}
\usepackage{graphicx}
\usepackage{stmaryrd}
\usepackage{flafter}

\usepackage{algorithm}
\usepackage{algpseudocode}
\usepackage{pifont}
\usepackage{varwidth}

\usepackage{flushend}

\allowdisplaybreaks

\newcommand{\subparagraph}{}
\usepackage[compact]{titlesec}
\titlespacing*{\section}{7pt}{0.8\baselineskip}{0.7\baselineskip}

\allowdisplaybreaks

\newcommand{\myhash}{%
  {\settoheight{\dimen0}{C}\kern-.05em\, \resizebox{!}{\dimen0}{\raisebox{\depth}{\#}}}}

\usepackage{caption}

\usepackage{subcaption}
\usepackage[numbers,sort&compress]{natbib}

\usepackage{multicol}


\input{symbols.tex}

\usepackage[none]{hyphenat}

\usepackage{tikz}
\usepackage{pgfplots,tikz-3dplot}
\pgfplotsset{compat=newest}
\usetikzlibrary{patterns}
\usetikzlibrary{positioning}
\usetikzlibrary{datavisualization}
\usetikzlibrary{datavisualization.formats.functions}
\usetikzlibrary{backgrounds}
\usetikzlibrary{shapes,snakes}

\usetikzlibrary{
	matrix
}

\tikzset{toprule/.style={%
		execute at end cell={%
			\draw [line cap=rect,#1] (\tikzmatrixname-\the\pgfmatrixcurrentrow-\the\pgfmatrixcurrentcolumn.north west) -- (\tikzmatrixname-\the\pgfmatrixcurrentrow-\the\pgfmatrixcurrentcolumn.north east);%
		}
	},
	bottomrule/.style={%
		execute at end cell={%
			\draw [line cap=rect,#1] (\tikzmatrixname-\the\pgfmatrixcurrentrow-\the\pgfmatrixcurrentcolumn.south west) -- (\tikzmatrixname-\the\pgfmatrixcurrentrow-\the\pgfmatrixcurrentcolumn.south east);%
		}
	}
}

 \usepackage{url}

\usepackage{afterpage}
\usepackage{balance}

\def\herm{{\sfH}}

\def\snr{{\mathsf{SNR}}}

\def\cg{{\clC\clN}}

\input{macros}

\author{\IEEEauthorblockN{Mahdi Nouri Boroujerdi\IEEEauthorrefmark{1}, Saeid Haghighatshoar\IEEEauthorrefmark{2}, \IEEEmembership{Member, IEEE}, Giuseppe Caire\IEEEauthorrefmark{2}, \IEEEmembership{Fellow, IEEE}}
	
	\IEEEauthorblockA{\IEEEauthorrefmark{1}{\small School of Electrical and Computer Engineering, University of Tehran,
		Tehran, Iran, Email: \{mhd.nouri\}@ut.ac.ir}}\\
	\IEEEauthorblockA{\IEEEauthorrefmark{2} {\small Communications and Information Theory Group, Technische Universit\"{a}t Berlin} \\[-10pt]
		{\small Email: \{saeid.haghighatshoar, caire\}@tu-berlin.de}}
	}

\title{Low-Complexity Statistically Robust Precoder/Detector Computation for Massive MIMO Systems\footnote{A short version of this paper was presented in the Workshop on Smart Antennas (WSA), Bochum, Germany, 2018 \cite{WSA_paper_published}.}}

\begin{document}

\maketitle

\vspace{-24pt} 

\begin{abstract}
Massive MIMO is a variant of multiuser MIMO in which the number of antennas at the base station (BS) $M$ is very large and typically much larger than the number of served users (data streams) $K$. Recent research  has widely investigated the system-level advantages of 
massive MIMO and, in particular, the beneficial effect of increasing the number of antennas $M$. 
These benefits, however, come at the cost of a dramatic increase in hardware  and  computational complexity. 
This is partly due to the fact that the BS needs to compute precoding/receiving vectors in order to coherently transmit/detect data to/from each user, 
where the resulting complexity grows proportionally to the number of antennas $M$ and the number of served users $K$. 
Recently, different algorithms based on tools from asymptotic random matrix theory and/or approximated message passing 
have been proposed to reduce such complexity. The underlying assumption in all these techniques, however, is that  the exact statistics (covariance matrix) of the channel vectors of the users is a priori known. This is far from being realistic, especially taking into account that, in the high-dim regime of $M \gg 1$, 
estimating the channel covariance matrices of the users is also challenging in terms of both computation and storage requirements. 
	In this paper, we propose a novel  technique for computing the precoder/detector in a massive MIMO system. 
	Our method is based on the randomized Kaczmarz algorithm and  does not require  a priori knowledge of the statistics of the users channel vectors. 
	We analyze the performance of our proposed algorithm theoretically and compare its performance with that of other techniques based on random matrix theory and approximate message passing via numerical simulations. Our results indicate that our proposed technique is computationally very competitive and yields quite a comparable  
	performance while it does not require the knowledge of the statistics of the users channel vectors. 
	\end{abstract}


\section{Introduction}
We consider a massive MIMO system where a base station (BS) is equipped with an array consisting of $M$  
antennas and serves $K$ users simultaneously using spatial multiplexing.  
In such a system, the spatial multiplexing gain of $K$ achieves a high sum spectral efficiency, while the large antenna diversity achieved by using 
$M \gg K$ antennas provides {\em channel hardening}, such that also the per-user rates are very stable and 
robust against statistical variations of the propagation environment \cite{Marzetta-TWC10,Intro_MM_Ashikhmin,Intro_MM_Rusek,Intro_MM_Larsson}. Using large number of antennas at the BS has many other interesting advantages: for instance, as the number of antennas tends to infinity, i.e. $M \to \infty$, the intra cellular interference vanishes so does the effect of noise and fast (small-scale) 
fading  \cite{Marzetta-TWC10}. Moreover, due to channel hardening and to the fact that for $M$ much larger than $K$ the users channel vectors tend to become almost mutually orthogonal, linear precoding/detection in the \textit{Uplink} (UL) and the \textit{Downlink} (DL) is sufficient to achieve a performance 
comparable to that of the optimal nonlinear transmitter/receiver \cite{Marzetta-TWC10,downlink_MM_linear,caire2003achievable,weingarten2006capacity}. 

In this paper, we focus on linear schemes such as  \textit{zero forcing} (ZF) precoder, also known as {\em decorrelator} (when used as BS receiver in the UL \cite{uplink_decorrelator,uplink_decorrelator_sync,uplink_decorrelator_async}), regularized ZF (RZF), where the regularization factor is tuned to 
the level of the noise at the users receivers, and \textit{linear minimum mean squared error} (MMSE) detector, 
which achieves the optimal Signal-to-Interference plus Noise Ratio (SINR)  among all linear receivers \cite{uplink_MMSE} (we refer  to \cite{MIMO_detection_survey} for a more detailed review of the available techniques for MIMO precoding and detection). 
Computing the corresponding linear precoding/detection matrices from the estimated users channel vectors 
requires matrix inversion and matrix-matrix multiplication, which requires $O(K^3)$ and $O(MK^2)$ operations, respectively. 
This imposes a very high computational complexity on massive MIMO implementations since the  number of BS antennas $M$ and the number 
of served users $K$ can be  very large. 
Several works have recently reported hardware chips that are designed to specifically compute beamforming matrices in massive MIMO, see e.g. \cite{MM-implementation} for a chip design for $M=128$ antennas and $K=8$ users. Many implementations of massive MIMO, however, are still using general purpose signal processing hardwares (such as DSP and FPGA). As a result,  a significant research effort has been recently devoted to developing efficient and low-complexity  algorithms for designing precoding/detection 
matrices in massive MIMO. Most of the works  exploit  underlying statistical properties of the channel states of the users and asymptotic results in random matrix theory to reduce the complexity. An important class of these algorithm use the \textit{truncated polynomial expansion} (TPE) technique \cite{Moshavi1996,muller2001design,Sessler2005, Zarei,Bjornson, 2017Benzin} in order to approximate the inverse matrix via a Taylor series expansion with a small number of terms. 
This reduces the complexity by a factor growing proportionally to the number of antennas $M$ (or users $K$). 
TPE-based techniques, however, require pre-calculation of the polynomial coefficients based on the statistics of the channel vectors, thus, require statistical information of the users channels. In particular, as seen in \cite{Zarei,Bjornson, 2017Benzin}, they require the knowledge of each users channel covariance matrix. 
This may not be practical, especially in massive MIMO regime where the covariance matrix of each user channel vector has very large dimension
$M \times M$ and needs to be estimated and tracked (i.e., updated over time) from the user channel measurements. Estimation and tracking of the users channel covariance matrix in the massive MIMO regime has also been intensely investigated in recent works (e.g., see \cite{haghighatshoar2017massive,haghighatshoar2016low}) and it is known to be quite a challenging and computationally demanding problem. 
A similar approach for linear MMSE detection consists of using a reduced rank filter, as in \cite{uplink_MMSE_reduced_rank}. 
In a recent line of work \cite{ghods2017optimally}, the \textit{approximate message passing} (AMP) algorithm, initially developed for sparse signal recovery in \cite{bayati2011dynamics}, has been used to reduce the computational complexity of massive MIMO precoding and detection.


In this paper, we propose a novel  technique for designing precoding/detection matrices in a massive MIMO system. 
Our method is based on the randomized Kaczmarz Algorithm (KA) and  does not require any knowledge of the statistics of the users' channels.
KA was initially proposed by Kaczmarz as an iterative technique for solving over-determined (OD) set of linear equations (SLE) \cite{Kaczmarz}. 
A randomized version of KA was recently proposed and analyzed for solving consistent OD SLE \cite{Strohmer2006}. 
With the advent of stochastic gradient techniques in machine learning, KA has been revitalized \cite{needell2014stochastic} and applied to other problems such as solving quadratic equations \cite{chi2016kaczmarz}. Our work in this paper is mainly based on the randomized KA proposed in \cite{Strohmer2006}. We extend this algorithm to the under-determined (UD) SLE  and use it to design new variants of KA for massive MIMO applications. 
We also analyze the performance of our proposed algorithm theoretically and compare its performance with that of other relevant techniques based on random matrix theory and approximate message passing via numerical simulations. Our results indicate that the proposed scheme has a comparable performance and is very competitive from a computational viewpoint,  while it does not require the knowledge of the statistics of the users channel vectors. 

\subsection{Notation}
We represent scalar constants by non-boldface letters {(e.g., $x$ or $X$)}, sets by calligraphic letters (e.g., $\Xc$),  vectors by boldface small letters (e.g., $\xv$), and matrices by boldface capital letters (e.g., $\Xm$). 
We denote the $i$-th row and the $j$-th column of a matrix $\bfX$ with the row-vector $\Xm_{i,.}$ and the column-vector $\Xm_{.,j}$ respectively. We represent the components of a vector $\bfx=(x_1, \dots, x_n)^\transp$ by $x_i$, $i=1,2, \dots, n$.
We denote a sequence of vectors by upper/lower indices, e.g., $\bfx^1, \bfx^2, \cdots$ or $\bfx_1, \bfx_2, \cdots$.
We indicate the Hermitian conjugate and the transpose of a matrix  $\bfX$ by $\bfX^\herm$ and $\Xm^\transp$ respectively, with the 
same notation being used for vectors and scalars. 
We denote the inner product between two matrices (and similarly two vectors) $\bfX$ and $\bfY$ 
by $\inp{\bfX}{\bfY}=\trace(\bfX^\herm \bfY)$. We use $\|\bfX\|_\sfF=\sqrt{\inp{\bfX}{\bfX}}$ for the Frobenius norm of a matrix $\bfX$ and $\|\bfx\|$ for the $l_2$-norm of a vector $\bfx$.
For two $m\times n$ and $m'\times n$ matrices $\bfX$ and $\bfY$, we denote by $[\bfX;\bfY]$ the $(m+m')\times n$ matrix obtained by stacking the rows of $\bfX$ on top of the rows of $\bfY$. 
The  identity matrix of order $p$ is represented by $\bfI_p$. For an integer $k>0$, we use the shorthand notation $[k]$ for $\{1,2,\dots, k\}$. 

\section{Problem Statement}\label{sec:definition}

Consider a multiuser massive MIMO system with a BS with $M\gg 1$ antennas serving $K$ single-antenna users.  We denote the \textit{true} channel vector of user $i \in [K]$ to the $M$ BS antennas with $\bfh^i \in \bC^M$, where $\bfh^i=(h^i_1, \dots, h^i_M)^\transp$ with $h^i_k$ denoting the complex channel gain to the  $k$-th antenna. We denote the \textit{true}  $M\times K$ channel matrix consisting of the channel vectors of all the users by $\bfH=[\bfh^1, \dots, \bfh^K]$. We also denote  a possibly quantized and noisy \textit{estimate} of this matrix available at the base station (BS) by the $M\times K$ matrix $\bfQ$. It is worthwhile to mention that, in this paper, we always assume that all the computations done at the BS are  based on the estimate matrix $\bfQ$.
In {\em time-division duplexing} (TDD) systems, the matrix $\bfQ$ is obtained during the training period of each channel 
coherence block\footnote{The channel coherence block is the region in the time-frequency plane where the small-scale fading process generating the coefficients of $\bfH$ can be considered constant. For a channel with physical coherence time $\Delta T_c$ and coherence bandwidth $\Delta W_c$, the channel coherence block spans approximately $\lceil \Delta T_c \times \Delta W_c \rceil$ complex signal dimensions.} 
by transmitting orthogonal pilots from the users to the BS in the UL (exploiting UL-DL channel reciprocity \cite{Marzetta-TWC10,shepard2012argos}).  
In {\em frequency-division duplexing} (FDD) systems, in contrast, the BS obtains the estimated channel matrix $\bfQ$  by 
 pilot transmission in the DL and some form of analog or quantized feedback 
from the users to the BS in the UL (see, e.g.,  \cite{caire2010multiuser}).


\subsection{Detection and precoding vectors  in the UL and the DL}\label{sec:up}

In the UL, all the $K$ users transmit their data symbols to the BS. The $M$-dim received vector at the BS  is given by:
\begin{equation}\label{eq:ul1}
\bfy=\bfH\bfs+\bfn,
\end{equation}
where \(\bfs=(s_1, \dots, s_K)^\transp \in \bC^K\) is the vector of symbols  sent by the $K$ users, where $\bfH$ is the true channel matrix of the users of size $M \times K$ as before,  and  where \(\bfn \sim \cg(0, \sigma^2 \bfI_M)\) is the additive zero-mean white Gaussian noise (AWGN) at the BS receiver.
Without loss of generality, we assume that the users symbols have mean zero and variance $\bE[|s_i|^2] = P$, and incorporate the pathloss effects due to possibly different distance of the users form the BS as part of the channel vectors $\{\bfh^k: k \in [K]\}$.  
As stated before,  we will consider only the case where the BS  detects symbols  \(\bfs\) from the observations  \(\bfy\) via a linear receiver. 
Two widely used linear detectors are zero forcing detector (ZFD) and minimum mean squared error detector (MMSED) which yield the following symbol estimates:
\begin{align}
\hat{\bfs}_\text{ZFD}&=(\bfQ^\herm\bfQ)^{-1}\bfQ^\herm\bfy\label{eq:ulZF},\\
\hat{\bfs}_\text{MMSED}&=(\bfQ^\herm\bfQ+\frac{1}{\SNR}\bfI_K)^{-1}\bfQ^\herm\bfy,\label{eq:ulMMSE}
\end{align}
where $\snr=\frac{P}{\sigma^2}$ denotes the per-user {\em transmit} \textit{Signal-to-Noise Ratio} (SNR).
In the case of massive MIMO, where $K$ is very large (e.g., some tens) and $M$ is even much larger (e.g., some hundreds),
even implementing these simple linear detectors is computationally challenging since they require  large-dim matrix-matrix multiplication 
which needs roughly $O\left( MK^2\right)$ operations, and matrix inversion which requires approximately $O\left(K^3\right)$ operations. Note that as an alternative to ZFD and MMSED in \eqref{eq:ulZF} and \eqref{eq:ulMMSE} respectively, one can apply  Maximum Ratio Combining (MRC) $\hat{\bfs}=\bfQ^\herm\bfy$ for detecting the users symbols in the UL. MRC  requires only matrix-vector multiplication and, hence, needs roughly $O(MK)$ operations, which is much less than that needed for ZFD or MMSED. However, in cases where high spectral efficiency is required (high-SNR conditions) or where inter-cell interference is not the dominating effect, MRC performs much worse than ZFD and MMSED in terms of achievable per-user rates, unless the number of antennas is very large \cite{RZF_Hoydis, Huh11, ZF_MRC_MM}
	\footnote{In cases where these conditions are not met, i.e., when the 
		system is dominated by inter-cell interference and/or when the CSI quality is poor due to high mobility, the advantage of ZFD/MMSED over MRC may disappear, 
		as shown for example in \cite[Section 6.3]{fundamental_MM}. In those cases, it is obvious that MRC becomes a simple and attractive alternative to any more sophisticated detection scheme.}.
%
%

A similar computational issue arises in the DL, where the BS transmits data to the users by coherently beamforming data to the users via suitable linear precoding matrices. We   consider  linear precoding schemes such as zero forcing beamforming (ZFBF) and regularized zero forcing beamforming (RZFBF), 
which have been shown to  achieve near-optimal performance in massive MIMO \cite{downlink_MM_linear}. The precoding vectors in the DL can be computed according to \cite{RZF_Hoydis,RZF_Peel} as:
\begin{align}
\mathbf{x}_{\text{ZFBF}}&=\beta_\text{ZFBF}\bfQ(\bfQ^\herm \bfQ)^{-1}\bfs\label{eq:dlZF}\\
\bfx_{\text{RZFBF}}&=\beta_\text{RZFBF} \bfQ(\bfQ^\herm \bfQ+\xi \bfI)^{-1}\bfs\label{eq:dlMMBF},
\end{align}
where $\beta_\text{ZFBF}$ and $\beta_\text{RZFBF}$ are scaling parameters to impose a given BS average transmit power, and where $\xi$ is the regularization parameter \cite{RZF_Peel}. Each user $k \in [K]$ will then receive the following signal
\begin{equation}
y_k={\bfh^k}^\herm\bfx+n_k,
\end{equation}
where $\bfx$ can  be either $\bfx_{\text{ZF}}$ or $\bfx_{\text{RZFBF}}$, and where $n_k\sim\cg(0, \sigma^2)$ is the AWGN at the user.
Computing precoding vectors of the ZFBF and RZFBF in \eqref{eq:dlZF} and \eqref{eq:dlMMBF} respectively has the same order of complexity 
as in the corresponding UL scenario \eqref{eq:ulZF} and \eqref{eq:ulMMSE}.

\subsection{Low-Complexity Beamforming via Kaczmarz Algorithm}\label{sec:KA_intro}
Our goal in this paper is to reduce the complexity of the detector/precoder computation in  the massive MIMO. 
We first show that detecting the users signals in the UL, i.e., $\hat{\bfs}_\text{ZFD}, \hat{\bfs}_\text{MMSED}$ in \eqref{eq:ulZF} and \eqref{eq:ulMMSE}, or finding the suitable $M$-dim precoding signal to be the transmitted to the users in the DL, i.e., $\mathbf{x}_{\text{ZFBF}}, \bfx_{\text{RZFBF}}$ in \eqref{eq:dlZF} and \eqref{eq:dlMMBF}, can be posed as  finding the solution $\bfw^*$ of an SLE of the form $\bfA \bfw=\bfb$ where  $\bfA$ depends on the estimated channel state $\bfQ$ and $\bfb$  depends on the users signal $\bfs$ in a DL  and on the noisy received  signal from the users $\bfy$ as in \eqref{eq:ul1} in a UL scenario. We find the optimal solution $\bfw^*$ of $\bfA \bfw=\bfb$ via an iterative procedure known as Kaczmarz Algorithm (KA). This technique was initially proposed by  Kaczmarz for efficiently solving a consistent overdetermined (OD) SLE \cite{Kaczmarz}. 

KA works as follows. At the start of each iteration $t$, the algorithm has an estimate $\bfw^t$ of the optimal solution $\bfw^*$ and selects one of the rows, say row number $r(t)$, of $\bfA$. If the estimate $\bfw^t$ satisfies the $r(t)$-th equation, namely, if $\inp{\bfa_{r(t)}}{\bfw^t}= b_{r(t)}$, where $\bfa_{r(t)}=(\bfA_{r(t),.})^\herm$ denotes the conjugate of the $r(t)$-th row of $\bfA$, then KA keeps $\bfw^t$ as it is. Otherwise, it updates $\bfw^t$ along $\bfa_{r(t)}$ to make the $r(t)$-th equation consistent. This is summarized in Algorithm \ref{tab:Kac}. Further insight into the performance of KA is gained by noting  that each Kaczmarz update indeed obeys the \textit{minimum energy perturbation} principle: At each iteration $t$, KA finds the closest vector $\bfw$ (in $l_2$-norm) to the current solution $\bfw^t$ that satisfies the selected equation or mathematically speaking $\bfw^{t+1}=\argmin_{\bfw} \|\bfw - \bfw^t\|^2 \text{ s.t. } \inp{\bfa_{r(t)}}{\bfw}= b_{r(t)}$.
%
%
%
%
\begin{algorithm}[t]
	\caption{ Kaczmarz Algorithm} 
	\label{tab:Kac} 
	\begin{algorithmic}[1]
		\State Initialize $\bfw^0$.
		\For{$t=0,1,\dots,T-1$}
		\State {Pick a  row $r(t)$ of $\bfA$ denoted by the row vector $\bfa_{r(t)}^\herm$. 
		
		\Comment{In \cite{Strohmer2006}, each row $i$ of the matrix $\bfA$ is selected randomly with the probability} $\frac{\|\bfa_i\|^2}{\|\bfA\|_\sfF^2}$. }
		\State Update $\bfw^t$ as $\bfw^{t+1}=\bfw^t+\frac{ b_{r(t)}- \inp{\bfa_{r(t)}}{ \bfw^t} }{\| \bfa_{r(t)}\|^2}\bfa_{r(t)}$. 
		
		 \Comment{After the
				update, $\inp{\bfa_{r(t)}}{\bfw^{t+1}}= b_{r(t)}$, and $r(t)$-th
				equation is fulfilled.} 
		\EndFor
	\end{algorithmic}
\end{algorithm}

The convergence speed of KA  is known to depend highly on the order of selecting the rows of $\bfA$, 
referred to as the {\em update schedule}. In particular, the algorithm can be quite slow under a round-robin  schedule that selects the rows periodically. This has motivated the study of randomized variants of KA in which the rows of $\bfA$ are selected randomly according to some probability distribution 
that might depend on $\bfA$.  Such a randomized version of KA was recently proposed and analyzed by Strohmer and Vershynin in \cite{Strohmer2006}, where they showed that the performance of the resulting randomized KA depends on the condition number of the matrix $\bfA$. 
It was also shown in  \cite{Strohmer2006} that, for a wide class of matrices, the proposed algorithm can potentially outperform other  
well-known techniques such as {\em conjugate gradient}. 

The traditional KA and also the randomized variant proposed in \cite{Strohmer2006} are applicable to OD but consistent cases, where there is a solution satisfying all the equations. For massive MIMO applications considered in this paper, in contrast, we will need variants of KA for both over-determined and under-determined scenarios. Moreover, in the over-determined cases the equations are almost always inconsistent. For example, as we will see, this inconsistency arises in a UL scenario due to the presence of noise in the received signal at the BS. In fact,  massive MIMO linear detection in such cases requires finding the least-squares solution of an appropriate  SLE. However,  an ad-hoc application of KA results in a residual term and does not yield the  desired solution \cite{Needell2010}. We investigate the effect of this residual term in the massive MIMO setup in Section \ref{sec:results:residual} via numerical simulations, where we illustrate that it generally incurs a considerable loss in per-user rate. Also, to avoid this effect, we develop a new variant of KA that always solves consistent equations after a suitable transformation, where we also show that the resulting updates of the new algorithm converge to the desired least-squares solution. Our method resembles \cite{extendedKaczmarz}, which uses a step prior to KA to remove the inconsistency of the equations but has much less computational complexity.


\section{Unified Mathematical Framework}\label{sec:math frame}
In this section, we propose a unified approach that we will use to analyze the performance of all variants of KA studied in the sequel.
We will focus on the SLE given by $\bfA \bfx=\bfb$ where $\bfA\in \bC^{m\times n}$ is the matrix of linear equations, where $\bfx \in \bC^n$ is the vector of unknown values, and where $\bfb\in \bC^m$ is the vector of known coefficients.  We say that the SLE $\bfA \bfx=\bfb$ is \textit{consistent} if it has a solution.  We always assume that $\bfA$ has full row-rank in the UD case ($m<n$) and full column-rank in the OD case ($m\geq n$). As we will see, 
these conditions are satisfied (with probability $1$) in  massive MIMO applications studied in this paper.

We will consider the randomized version of KA explained in Algorithm\,\ref{tab:Kac}, where at each iteration the algorithm selects 
a row of  $\bfA$ randomly and independently from the previous iterations and according to a given probability distribution $\bfp=(p_1, \dots, p_m)^\transp$, where $p_i\in [0,1]$ is the probability of selecting row $i\in [m]$ and where $\sum_{i\in [m]}p_i=1$. We denote the random rank-1 projection operator produced by the random row selection by $\scrP_R=\frac{\bfa_R \bfa_R^\herm}{\|\bfa_R\|^2} \in \bC^{n\times n}$, where $R \in [m]$ is a random variable denoting the index of the selected row, where $\bP[R=i]=p_i$ and where we denote by  $\bfa_i=(\bfA_{i,.})^\herm \in \bC^m$ the conjugate of the $i$-th row of $\bfA$. We define the average of the projection  operator $\scrP_R$ as 
\begin{align}\label{avg_proj}
\overline{\scrP}=\bE_R[\scrP_R]=\sum_{i=1}^m p_i \frac{\bfa_i \bfa_i^\herm}{\|\bfa_i\|^2}.
\end{align}
Note that $\overline{\scrP}$ is an $n \times n$ \textit{positive semi-definite} (PSD) matrix, with $\tr(\overline{\scrP})=\sum_{i \in [m]}p_i=1$.
As we will see in the next sections:

(1) We will design KA such that, depending on the specific scenario, all the estimates generated across all the iterations of KA  belong to the subspace  $\clX \subset \bC^n$ generated by conjugate of the rows of  $\bfA$, i.e., the column span of $\bfA^\herm$. 

(2) The crucial parameter that will control the convergence speed of KA will be the average gain of the matrix $\bfA$ over the subspace $\clX$ defined  as follows.
\begin{definition}[average gain]\label{avg_gain}
	Let $\bfA$ be an $m\times n$ matrix.  Let $\bfp \in \bR_+^m$ be a probability distribution over the rows of $\bfA$ and define the average projection operator $\overline{\scrP}$ produced by the rows of $\bfA$ as in \eqref{avg_proj}. Let $\clX$ be the subspace produced by the span of columns of $\bfA^\herm$. The average gain of the matrix $\bfA$ over the subspace $\clX$ is defined as:
	\begin{align}
	\kappa_{\clX}(\bfA, \bfp)=\min_{\bfx\in \clX, \bfx \neq 0} \frac{\bfx^\herm \overline{\scrP} \bfx}{\|\bfx\|^2}.
	\label{eq:opt_gain}
	\end{align}
\end{definition}
\input{optimal_p_plot2.tex}

Note that $\kappa_{\clX}(\bfA, \bfp)$ depends both on the matrix $\bfA$  and the probability distribution $\bfp$. Moreover, as $\tr(\overline{\scrP})=1$ and as the subspace $\clX$ has dimension $\min\{m,n\}$ and $\overline{\scrP}$ acts on this subspace, we have that $\kappa_{\clX}(\bfA, \bfp) \in [0,\frac{1}{\min\{m,n\}}]$. 
We will show that the closer $\kappa_\clX(\bfA, \bfp)$  is to $\frac{1}{\min\{m,n\}}$, the faster KA converges across the iterations. 
Note that from 
\eqref{eq:opt_gain}, for a given matrix $\bfA$, $\kappa_{\clX}(\bfA, \bfp)=\lambda_{\min}(\overline{\scrP}| \clX)$ is given by the minimum eigen-value of $\overline{\scrP}$ over the subspace $\clX$. 
As $\overline{\scrP}$ is a linear function of $\bfp$ and as $\lambda_{\min}(\overline{\scrP}| \clX)$ is a concave function of $\overline{\scrP}$ \cite{tao2012topics}, $\kappa_{\clX}(\bfA, \bfp)$ is a concave function of $\bfp$ over the convex set of probability distributions over the row set of $\bfA$. Therefore, finding the best distribution $\bfp$ maximizing $\kappa_{\clX}(\bfA, \bfp)$ can be posed as maximizing a concave function over a convex set. This can be formulated as a \textit{semi-definite programming} (SDP) and solved with an affordable complexity via convex optimization techniques. However, this typically incurs the same order of complexity as directly computing the detection/precoding vectors and is not suitable for the massive MIMO applications  addressed in this paper. This might be, however, very useful for situations where one solves many equations with the same matrix $\bfA$, whereby the complexity of computing the optimal $\bfp$ is amortized over time. This is not the case in massive MIMO applications since the matrix $\bfA$ depends on the estimated channel matrix $\bfQ$, which changes from one coherence time to another.  Instead, we will use the suboptimal probability distribution $\bfp$ proposed by \cite{Strohmer2006}, where $p_i=\frac{\|\bfa_i\|^2}{\|\bfA\|_\sfF^2}$ and the probability of selecting each row $i\in [m]$ scales proportionally to its  squared  $l_2$-norm $\|\bfa_i\|^2$. Note that computing this 
probability distribution incurs a complexity of $O(mn)$, which is definitely much lower than directly solving an SDP to find the optimal $\bfp$. 
Replacing this probability distribution in \eqref{avg_proj} results in the following 
average operator:
\begin{align}\label{avg_proj2}
\overline{\scrP}&=\bE_R[\scrP_R]=\sum_{i=1}^m p_i \frac{\bfa_i \bfa_i^\herm}{\|\bfa_i\|^2}=\sum_{i=1}^m \frac{\|\bfa_i\|^2}{\|\bfA\|_\sfF^2} \frac{\bfa_i \bfa_i^\herm}{\|\bfa_i\|^2}=\sum_{i=1}^m \frac{\bfa_i \bfa_i^\herm}{\|\bfA\|_\sfF^2}=\frac{\bfA^\herm \bfA}{\|\bfA\|_\sfF^2}.
\end{align}
From  Definition \ref{avg_gain}, this yields the following gain
parameter for the matrix $\bfA$, which will be used in the rest of the paper.
\begin{definition}\label{min_gain}
	Let $\bfA$ be the $m\times n$ matrix as defined before and let $\clX$ be a linear subspace of $\bC^n$ produced by the column span of $\bfA^\herm$. We define the normalized minimum gain of the matrix $\bfA$ along the subspace $\clX$ as: 
	\begin{align}
	\kappa_\clX(\bfA)=\min_{\bfx\in \clX, \bfx \neq 0} \frac{\|\bfA \bfx\|^2}{\|\bfA\|_\sfF^2 \|\bfx\|^2},
	\end{align}
\hfill $\lozenge$
	\end{definition}

	Note that $\kappa_\clX(\bfA) \in [0,\frac{1}{\min\{m,n\}}]$ as before, and the closer $\kappa_\clX(\bfA)$ becomes to $\frac{1}{\min\{m,n\}}$, the faster KA converges across the iterations. Fig.\,\ref{fig:opt_subopt} compares the sub-optimal gain $\kappa_\clX(\bfA)$ with the best average gain $\kappa_{\clX}(\bfA, \bfp^*)$ obtained by the optimal probability distribution $\bfp^*$. We consider a massive MIMO scenario where  the matrix $\bfA$ has  i.i.d. $\cg(0,1)$ components with $M=256$ rows corresponding to BS antennas but different number of columns $K\in \{10,15,25\}$ corresponding to the number of served users. KA for such a matrix $\bfA$ arises in an UL scenario where the users have spatially white channel vectors. 
From well-known results in random matrix theory \cite{random_mtx_verdu}, it is not difficult to show that with a very high probability (on the realization of the matrix $\bfA$) 
\begin{align}
\kappa_\clX(\bfA)\approx\frac{(\sqrt{M}-\sqrt{K})^2}{MK}\label{dumm_rev}
=\frac{(1-\sqrt{\rho})^2}{K},
\end{align}
where $\rho:=\frac{K}{M}$ denotes the ratio between the number of served users $K$ and the number of BS antennas $M$, typically called the \textit{loading factor} of the massive MIMO system.
Considering the fact that $\kappa_{\clX}(\bfA, \bfp^*)\leq \frac{1}{\min\{M,K\}}=\frac{1}{K}$, it is seen from \eqref{dumm_rev} that, for a fixed loading factor $\rho$, the matrix gain $\kappa_\clX(\bfA)$ gives an approximation of the best gain $\kappa_{\clX}(\bfA, \bfp^*)$ up to a  multiplicative factor (see also Fig.\,\ref{fig:opt_subopt}).  
	
This, as we will see, results in a KA that converges slower than the optimally-tuned KA up to  the same multiplicative factor.
	
	For the rest of the paper, we will always use the suboptimal row distribution, which can be computed in $O(mn)$ iterations.
	We now prove the following result that gives a unified performance guarantee for all the cases that will be considered in this paper.
	\begin{theorem}\label{main_thm}
		Let $\bfA$ be an $m \times n$ matrix  and let $\bfb \in \bC^m$. Let $\clX$ be the subspace of $\bC^n$ generated by the column span of $\bfA^\herm$ and suppose that $\bfA \bfx = \bfb$ has a solution $\bfx^* \in \clX$. Let $\bfx^0 \in \clX$ be an arbitrary initialization 
		  and let $\bfx^t$ be the estimate at time $t$  of the randomized KA with the suboptimal distribution as explained before starting from $\bfx^0$.  Then,
		\begin{align}
		\bE\big [\|\bfx^t - \bfx^*\|^2\big ] \leq (1-\kappa_\clX(\bfA))^t \|\bfx^0 - \bfx^*\|^2,
		\end{align}
		where the expectation is taken over the randomized row selection in KA, and where $\kappa_\clX(\bfA)$ denotes the minimum gain of the linear operator $\bfA$ over the subspace $\clX$ as in Definition \ref{min_gain}. \hfill $\square$
		\end{theorem}
		
		\begin{proof}
			From the update equation of KA, we have:
			\begin{align}
			\bfx^{t+1}- \bfx^*&=\bfx^t - \frac{\inp{\bfa_{r(t)}}{\bfx^t} - b_{r(t)}}{\|\bfa_{r(t)}\|^2} \bfa_{r(t)} - \bfx^*=\bfx^t - \bfx^* - \frac{\inp{\bfa_{r(t)}}{\bfx^t - \bfx^*}}{\|\bfa_{r(t)}\|^2} \bfa_{r(t)}\nonumber\\
			&=\big(\bfI_n - \frac{\bfa_{r(t)} \bfa_{r(t)}^\herm }{\|\bfa_{r(t)}\|^2}\big) (\bfx^t - \bfx^*)=(\bfI_n-\scrP_{r(t)})(\bfx^t - \bfx^*),
			\end{align}
			where $r(t) \in [m]$ denotes the index of the random row selected by KA at iteration $t$, where $b_{r(t)}$ is the component of $\bfb$ at position $r(t)$, where we used the fact that $\bfx^*$ is a solution of $\bfA \bfx=\bfb$, thus, $b_{r(t)}=\inp{\bfa_{r(t)}}{\bfx^*}$, and where $\scrP_{r(t)}:=\frac{\bfa_{r(t)} \bfa_{r(t)}^\herm}{\|\bfa_{r(t)}\|^2}$ is the random rank-1 projection operator onto the subspace spanned by $\bfa_{r(t)}$ as defined before. It is not difficult to check that $\bfI_n-\scrP_{r(t)}$ is also a projection operator, thus, $(\bfI_n-\scrP_{r(t)})^\herm (\bfI_n-\scrP_{r(t)})=(\bfI_n-\scrP_{r(t)})$. This yields: 
			\begin{align}
			\|\bfx^{t+1} - \bfx^*\|^2=(\bfx^{t}-\bfx^*)^\herm (\bfI_n-\scrP_{r(t)}) (\bfx^{t}-\bfx^*).
			\end{align}
			Note that  $\bfx^{t}-\bfx^*$  is independent of $\bfI_n-\scrP_{r(t)}$ since it depends on the random selection of the rows at iterations before $t$. Using this independence, conditioning on all row selections before $t$, and taking the expected value over $r(t)$ yields: 
			\begin{align}
			\bE\Big [&\|\bfx^{t+1}-\bfx^*\|^2 \big |r(1), \dots, r(t-1)\Big ]
			=(\bfx^t - \bfx^*)^\herm \bE[\bfI_n - \scrP_{r(t)}] (\bfx^t - \bfx^*)\\
			&=(\bfx^t - \bfx^*)^\herm (\bfI_n - \overline{\scrP}) (\bfx^t - \bfx^*)
			= \|\bfx^t - \bfx^*\|^2 - (\bfx^t - \bfx^*)^\herm \overline{\scrP} (\bfx^t - \bfx^*)\\
			&\stackrel{(a)}{\leq} (1-\kappa_\clX(\bfA))\|\bfx^t - \bfx^*\|^2,
			\end{align}
			where in $(a)$ we used the fact that the minimum gain of $\overline{\scrP}$ for the mentioned suboptimal distribution is given by $\kappa_\clX(\bfA)$ as in Definition \ref{min_gain}.
			Finally, taking expected value with respect to $r(1), \dots, r(t-1)$ and applying the induction we obtain:
			\begin{align}
			\bE\big[\|\bfx^{t+1}-\bfx^*\|^2\big]& \leq \big(1-\kappa_\clX(\bfA) \big)\bE\big[\|\bfx^{t}-\bfx^*\|^2\big] \leq \big(1-\kappa_\clX(\bfA) \big)^{t+1} \|\bfx^0 - \bfx^*\|^2,
			\end{align}
			where we used the fact that the initial point $\bfx^0$ is selected deterministically and dropped the expectation. This completes the proof.
			\end{proof}

			\begin{remark}
				From Theorem \ref{main_thm}, it is seen that the convergence speed of KA depends on the parameter $\kappa_\clX(\bfA)$, where a larger $\kappa_\clX(\bfA)$ guarantees a faster convergence. Moreover, since $\kappa_\clX(\bfA) \leq \frac{1}{\min\{m,n\}}$, at least $O(\min\{m,n\})$  iterations are needed for the convergence.  \hfill $\lozenge$
				\end{remark}
				
				\begin{remark}
					Since the random variable $\|\bfx^t - \bfx^*\|^2$ is non-negative, from the well-known Markov's inequality \cite{grimmett2001probability}, its expected value also provides a probabilistic bound on its tail behavior, i.e., for any $\zeta>0$, we have: 
					\begin{align}
					\bP\big [\|\bfx^t - \bfx^*\|^2 \geq \zeta\big] &\leq \frac{\bE\big[\|\bfx^{t}-\bfx^*\|^2\big]}{\zeta} \leq \big(1-\kappa_\clX(\bfA) \big)^t \frac{\|\bfx^0 - \bfx^*\|^2}{\zeta},
					\end{align}
					where $\kappa_\clX(\bfA)$ is as defined before.  \hfill $\lozenge$
					\end{remark}
										In the following, we will design several variants of KA depending on the specific beamforming scenario and apply Theorem \ref{main_thm} to analyze the convergence speed of the resulting algorithm. We will also explain how the convergence speed depends on the system parameters such as the number of antennas $M$, the number of served users $K$, and the spatial correlation of the channel vectors of the users.
					

\section{Kaczmarz Algorithm for different Beamforming Scenarios}\label{sec:scenarios}

In this section, we consider a multiuser massive MIMO system with a BS having $M$ antennas and $K$ single-antenna users. We remind from Section \ref{sec:definition} that we denote the perfect channel state of the users by an $M\times K$ matrix $\bfH$ and the noisy and possibly quantized estimate of $\bfH$  available at the BS by the $M\times K$ matrix $\bfQ$. 

\subsection{MMSE/ZF detection in the UL}\label{sec:mmse_ul}

We first consider a UL scenario. From what we have previously discussed in Section \ref{sec:up}, the ZFD in \eqref{eq:ulZF} and MMSED in \eqref{eq:ulMMSE} can be derived, by setting $\xi=0$ and $\xi=\frac{1}{\snr}$ respectively, from the following general expression:
\begin{align}\label{MMSED_eq}
\widehat{\bfs}=(\bfQ^\herm \bfQ + \xi \bfI_K)^{-1} \bfQ^\herm \bfy.
\end{align}
Let us now focus on the general case where $\xi\ne0$.

 We make the following observation, that we prove here for the sake of completeness. 
\begin{proposition}\label{prop:mmse_bf_ul}
	Let $\bfQ$, $\bfy$ and $\widehat{\bfs}$ be as in \eqref{MMSED_eq}. Then, the signal estimate $\widehat{\bfs}$ is the optimal solution of the optimization problem  $\argmin_{\bfw \in \bC^K} \|\bfQ \bfw - \bfy\|^2 + \xi \|\bfw\|^2$. \hfill $\square$
\end{proposition}
\begin{proof}
	Taking the derivative of the cost function gives $2\bfQ^\herm (\bfQ \bfw - \bfy) + 2\xi \bfw$. Setting this equal to zero yields the optimal solution $\bfw^*=(\bfQ^\herm \bfQ + \xi \bfI_K)^{-1} \bfQ^\herm \bfy$ which coincides with $\widehat{\bfs}$. 
\end{proof}

Note that we can write the $l_2$ regularized least-squares cost function $\|\bfQ \bfw - \bfy\|^2 + \xi \|\bfw\|^2$ in Proposition \ref{prop:mmse_bf_ul}  as $\|\bfB \bfw - \bfy_0\|^2$, where $\bfB=[\bfQ; \sqrt{\xi} \bfI_K]$ is an $(M+K)\times K$ matrix and where $\bfy_0=[\bfy; 0]$ is the $(M+K)$-dim vector obtained by appending a $K$-dim zero vector to the noisy received signal $\bfy$. The signal estimate $\widehat{\bfs}$ in 
\eqref{MMSED_eq}
 can also be written as the least-squares solution of the newly defined cost function and is given by
\begin{align}\label{s_eq_mmse}
\widehat{\bfs}=(\bfB^\herm \bfB)^{-1} \bfB^\herm \bfy_0= (\bfQ^\herm \bfQ + \xi \bfI_K)^{-1} \bfQ^\herm \bfy.
\end{align}

Let us consider the SLE $\bfB \bfw =\bfy_0$. Note that this SLE is OD and is always inconsistent when $\xi>0$. 
Furthermore, it is generally inconsistent even when $\xi=0$ due to the presence of noise in the received measurements $\bfy$. 
Thus, a direct application of the randomized KA will result in a residual error as we discussed in Section \ref{sec:KA_intro} and as we will illustrate via numerical simulations in Section \ref{sec:results:residual}.  To avoid this, we solve this SLE in two steps. We first define a new vector 
\begin{align}\label{y0_MMSED}
\widehat{\bfy}_0= \bfB \widehat{\bfs}\stackrel{(i)}{=}\bfB (\bfB^\herm \bfB)^{-1}\bfB^\herm \bfy_0=\bfB (\bfB^\herm \bfB)^{-1}\bfQ^\herm \bfy,
\end{align}
where in $(i)$ we used \eqref{s_eq_mmse}.
Note that $\widehat{\bfy}_0$ lies in the subspace $\clX$ of $\bC^{M+K}$ spanned the columns of $\bfB$. Moreover, from
\begin{align}\label{mmse_eq_ul_dumm}
\bfB^\herm \widehat{\bfy}_0=(\bfB^\herm \bfB) (\bfB^\herm \bfB)^{-1}\bfQ^\herm \bfy=\bfQ^\herm \bfy,
\end{align}
it is seen that $\widehat{\bfy}_0$ satisfies the UD SLE $\bfB^\herm \bfz=\bfb$ for the unknown $\bfz \in \bC^{M+K}$ and the set of known coefficients $\bfb=\bfQ^\herm \bfy$. We  first apply the randomized KA to this SLE to recover $\widehat{\bfy}_0$. We then have the following result.
\begin{proposition}\label{prop:MMSED_ul_2}
	Let $\bfz^t$ be the estimate obtained at  iteration $t$ of KA applied to the UD SLE $\bfB^\herm \bfz=\bfb$ where $\bfb=\bfQ^\herm \bfy$. Assume that KA starts from the zero initialization $\bfz^0=0$ and let $\widehat{\bfy}_0$ be as in \eqref{y0_MMSED}. Then
	\begin{align}
	\bE\big [\|\bfz^t - \widehat{\bfy}_0\|^2 \big ] \leq (1-\kappa_\clX(\bfB^\herm))^t \|\widehat{\bfy}_0\|^2,
	\end{align}
	where $\kappa_\clX(\bfB^\herm)$ is the minimum gain of the matrix $\bfB^\herm$ over the subspace $\clX$ generated by the columns of $\bfB$. \hfill $\square$
\end{proposition}
\begin{proof}
	From the update equation for KA and applying a simple induction, it is seen that starting from $\bfz^0=0$, all the estimates $\bfz^t$ produced by KA lie in the subspace $\clX$ produced by the columns of $\bfB$. Thus, the result simply follows by applying Theorem \ref{main_thm}.
\end{proof}

Once $\widehat{\bfy}_0$ is recovered, we solve the OD SLE $\bfB \widehat{\bfs}=\widehat{\bfy}_0$ in \eqref{y0_MMSED} to find the estimate $\widehat{\bfs}$. In contrast with the initial SLE $\bfB \widehat{\bfs}=\bfy$, which was inconsistent, the OD SLE in \eqref{y0_MMSED} is always consistent since $\widehat{\bfy}_0$ always lies in the subspace spanned by the columns of $\bfB$. Note that since $\bfB \widehat{\bfs}=\widehat{\bfy}_0$ is consistent and $\bfB=[\bfQ; \sqrt{\xi}\bfI_K]$ has $\sqrt{\xi}\bfI_K$ as its submatrix, $\widehat{\bfs}$ is simply given by the last $K$ components of $\widehat{\bfy}_0$ divided by $\sqrt{\xi}$. As a result, there is no need to solve the second SLE $\bfB \widehat{\bfs}=\widehat{\bfy}_0$ per se. 
Moreover, denoting by $\bfz^t$ the estimate of KA for the first SLE $\bfB^\herm \bfz=\bfb$ with $\bfb=\bfQ^\herm \bfy$, we can obtain an estimate $\widehat{\bfs}^t$ at each iteration $t$ of KA by dividing the last $K$ components of $\bfz^t$ by $\sqrt{\xi}$. By introducing $\bfu^t\in \bC^M$ and $\bfv^t\in \bC^K$ such that $\bfz^t=[\bfu^t; \sqrt{\xi} \bfv^t]$, we can summarize KA for $\bfB^\herm \widehat{\bfy}_0=\bfQ^\herm \bfy$ as in Algorithm \ref{tab:Kac_MMSED}. A direct inspection shows that, due to the submatrix $\sqrt{\xi} \bfI_K$ in $\bfB=[\bfQ; \sqrt{\xi} \bfI_K]$, at each iteration of KA only a single  randomly selected component of the estimate $\bfv^t$ is updated, where $\bfv^t$ converges to the least-squares estimate of the symbols $\widehat{\bfs}$ corresponding to the MMSED. In particular, due to the continuity of $\widehat{\bfs}$ with respect to $\xi$ as $\xi \to 0$, one can obtain the ZFD estimate $\widehat{\bfs}$ by setting $\xi=0$  in Algorithm \ref{tab:Kac_MMSED}.

\begin{algorithm}[t]
	\caption{ Kaczmarz Algorithm for the UL} 
	\label{tab:Kac_MMSED} 
	\begin{algorithmic}[1]
		\State {\bf Input:} {Channel state estimate} $\bfQ\in \bC^{M \times K}$, received noisy UL signal  $\bfy \in \bC^M$, and MMSED parameter $\xi\geq 0$. \Comment{$\xi=0$ corresponds to the ZFD}.
		\State Compute $\bfb=\bfQ^\herm \bfy \in \bC^K$.
		\State Define $\bfu^t \in \bC^M$ and $\bfv^t \in \bC^K$ with $\bfu^0=0$, $\bfv^0=0$.  
		\For{$t=0,1,\dots, T-1$}
		\State Pick a  row $r(t)$ of $\bfQ^\herm$ denoted by $\bfq_{r(t)}^\herm$ with a probability proportional to $\|\bfq_{r(t)} \|^2+ \xi$ given by $\frac{\|\bfq_{r(t)} \|^2+ \xi}{\|\bfQ\|_\sfF^2 + K \xi}$.
		\State Compute the residual $\gamma^t:=\frac{ b_{r(t)}- \inp{\bfq_{r(t)}}{\bfu^t} - \xi\, v^{t}_{r(t)}}{\|  \bfq_{r(t)}\|^2 + \xi}$.
		\State Update $\bfu^{t+1}=\bfu^t + \gamma^t \bfq_{r(t)}$.
		\State Update $v_{r(t)}^{t+1}=v^t_{r(t)}+ \gamma^t$, and $v_{j}^{t+1}=v^t_{j}$ for $j \neq r(t)$.
		\EndFor
				\State {\bf Output:}  Set $\bfw=\bfv^{T-1}$.
		\State {\bf Output:} Set the signal estimate to $\widehat{\bfs}=\bfw$.
	\end{algorithmic}
\end{algorithm}

To analyze the convergence speed of KA, from Proposition \ref{prop:MMSED_ul_2}, we need to compute the gain of the matrix $\bfB^\herm$ along the subspace $\clX$ spanned by the conjugate of the rows of $\bfB^\herm$, that is, the column span of $\bfB$. 
Since every vector in $\clX$ can be written as $\bfB \bfw$ for some $\bfw\in \bC^K$, a direct calculation shows that $\kappa_\clX(\bfB^\herm)$ is given by: 
\begin{align}
\kappa_\clX(\bfB^\herm)&=\min_{\bfx \in \clX, \bfx\neq 0} \frac{\|\bfB^\herm \bfx\|^2}{\|\bfB^\herm\|_\sfF^2 \|\bfx\|^2}=\min_{\bfw \in \bC^K, \bfw\neq 0} \frac{\|\bfB^\herm \bfB \bfw\|^2}{\|\bfB^\herm\|_\sfF^2 \|\bfB \bfw\|^2}\nonumber\\
&=\frac{\lambda_{\min} (\bfB^\herm \bfB)}{\|\bfB\|_\sfF^2}=\frac{\lambda_{\min} (\bfQ^\herm \bfQ) + \xi}{ \|\bfQ\|_\sfF^2 + K\xi}.\label{MMSED_DL}
\end{align}
Note that $\frac{\lambda_{\min} (\bfQ^\herm \bfQ) }{ \|\bfQ\|_\sfF^2} \leq \frac{1}{K}$ and it is not difficult to check that for any $\xi\ge0$, we have
\begin{align}
 \frac{\lambda_{\min} (\bfQ^\herm \bfQ) }{ \|\bfQ\|_\sfF^2} \leq \frac{\lambda_{\min} (\bfQ^\herm \bfQ) + \xi}{ \|\bfQ\|_\sfF^2 + K\xi}  \leq \frac{1}{K}.
\end{align}
This implies that the algorithm converges faster for larger values of $\xi\ge0$ (i.e., smaller UL SNR) and the slowest convergence corresponds to the ZFD with $\xi=0$.
 When the channel vectors of the users have white complex Gaussian components (no spatial correlation),  from well-known results in random matrix theory \cite{random_mtx_verdu}, it results that 
\begin{align}
\kappa_\clX(\bfB^\herm)&\geq  \frac{\lambda_{\min} (\bfQ^\herm \bfQ)}{\|\bfQ\|_\sfF^2}= \frac{(\sqrt{M}-\sqrt{K})^2}{MK}
=\frac{(1-\sqrt{\rho})^2}{K}.
\end{align}
Therefore, from Proposition \ref{prop:MMSED_ul_2} it results that $O(\frac{K}{(1-\sqrt{\rho})^2})$ iterations are sufficient to reasonably estimate the symbol sequence $\widehat{\bfs}$ for  the MMSED for all ranges of $\xi> 0$ and for the ZFD for $\xi=0$. In particular, it is seen that the required estimation time grows linearly with the number of users $K$. Since each step of KA requires multiplying two $M$-dim vectors, thus, requiring $M$ operations, the total computational complexity scales as $O(\frac{KM}{(1-\sqrt{\rho})^2})$. To this we should also add $O(KM)$ operations needed for computing the sub-optimal row-sampling probability distribution and $O(KM)$ operations needed for calculating $\bfQ^\herm \bfy$ in Algorithm \ref{tab:Kac_MMSED}.

\subsection{RZF/ZF Precoding in the DL}\label{sec:MMBF_dl}

In this section, we consider  RZF precoding in the DL. Suppose $\bfs=(s_1, \dots, s_K)^\transp$ is the symbol sequence BS intends to send to the users in the DL.  In the RZFBF in the DL, the BS transmits the $M$-dim vector $\bfy:=\bfQ(\bfQ^\herm \bfQ + \xi \bfI_K)^{-1} \bfs$ via its $M$ antennas. 
We  observe that $\bfy$ lies in the subspace spanned by the columns of $\bfQ$. We  design a new KA that finds the $K$-dim signal $\bfw=(\bfQ^\herm \bfQ + \xi \bfI_K)^{-1}\bfs$, from which we can obtain $\bfy=\bfQ \bfw$ via a matrix-vector multiplication. We first prove the following result.
\begin{proposition}\label{prop:mmbf_dl}
	Let $\bfQ$, $\bfs$ and $\bfw=(\bfQ^\herm \bfQ + \xi \bfI_K)^{-1}\bfs$ be as introduced before. Then, $\bfw$ is the optimal solution of the following optimization problem 
	$\bfw=\argmin_{\bfx \in \bC^K}  \|\bfQ \bfx\|^2+ \xi \|\bfx - \bfs_\xi\|^2 $, where $\bfs_\xi=\frac{\bfs}{\xi}$. \hfill $\square$
\end{proposition}
\begin{proof}
	Taking the derivative with respect to $\bfx$ and setting it equal to $0$ yields $\bfQ^\herm \bfQ \bfx + \xi \bfx - \bfs=0$. This gives the optimal solution as $\bfx^*=(\bfQ^\herm \bfQ + \xi \bfI_K)^{-1} \bfs$, which coincides with $\bfw=(\bfQ^\herm \bfQ + \xi \bfI_K)^{-1}\bfs$.
\end{proof}

\begin{algorithm}[t]
	\caption{ Kaczmarz Algorithm for the DL} 
	\label{tab:Kac_MMBF} 
	\begin{algorithmic}[1]
		\State {\bf Input:} Channel state estimate $\bfQ\in \bC^{M \times K}$, users symbols $\bfs \in \bC^K$, and RZFBF parameter $\xi\geq 0$. \Comment{$\xi=0$ corresponds to the ZFBF}.
		\State Define $\bfu^t \in \bC^M$ and $\bfv^t \in \bC^K$ with $\bfu^0=0$, $\bfv^0=0$.  
		\For{$t=0,1,\dots, T-1$}
		\State Pick a  row $r(t)$ of $\bfQ^\herm$ denoted by $\bfq_{r(t)}^\herm$ with a probability proportional to $\|\bfq_{r(t)} \|^2+ \xi$ given by $\frac{\|\bfq_{r(t)} \|^2+ \xi}{\|\bfQ\|_\sfF^2 + K \xi}$.
		\State Compute the residual $\gamma^t:=\frac{ s_{r(t)}- \inp{\bfq_{r(t)}}{\bfu^t} - \xi\, v^{t}_{r(t)}}{\|  \bfq_{r(t)}\|^2 + \xi}$.
		\State Update $\bfu^{t+1}=\bfu^t + \gamma^t \bfq_{r(t)}$.
		\State Update $v_{r(t)}^{t+1}=v^t_{r(t)}+ \gamma^t$, and $v_{j}^{t+1}=v^t_{j}$ for $j \neq r(t)$.
		\EndFor
		\State {\bf Output:}  Set $\bfw=\bfv^{T-1}$.
		\State {\bf Output:} Set the desired DL signal  to $\bfy=\bfQ \bfw$.
	\end{algorithmic}
\end{algorithm}

Proposition \ref{prop:mmbf_dl} implies that we can obtain $\bfw=(\bfQ^\herm \bfQ + \xi \bfI_K)^{-1}\bfs$ as the optimal solution of $\|\bfQ \bfx\|^2 + \xi \|\bfx - \bfs_\xi\|^2$. We first write this as the least-squares minimization $\|\bfB \bfx-\bfb\|^2$ where $\bfB=[\bfQ; \sqrt{\xi}\bfI_K]$ is an $(M+K)\times K$ matrix and where $\bfb=[0;\sqrt{\xi}\bfs_\xi]$ is an $(M+K)$-dim vector with $\sqrt{\xi}\bfs_\xi=\frac{\bfs}{\sqrt{\xi}}$ in its last positions. 
The SLE $\bfB \bfw=\bfb$ is  OD  and should be solved for the vector $\bfw \in \bC^K$, from which we obtain the desired signal $\bfy=\bfQ \bfw$ to be transmitted in the DL. 
This SLE is inconsistent unless $\bfs=0$. Thus, we follow the same procedure as in the design of MMSED for the UL. 
We first define $\bfz=\bfB \bfw$ and solve the UD SLE 
\begin{align}\label{mmbf_z_dumm}
\bfB^\herm \bfz= \bfB^\herm \bfb=\sqrt{\xi} \frac{\bfs}{\sqrt{\xi}}=\bfs,
\end{align}
 for $\bfz$ and then solve the OD but consistent SLE $\bfB \bfw=\widehat{\bfz}$, where $\widehat{\bfz}$ is the estimate of $\bfz$ obtained from KA applied to \eqref{mmbf_z_dumm}.
Note that, as in the MMSED, due to the presence of $\sqrt{\xi} \bfI_K$ submatrix in $\bfB=[\bfQ; \sqrt{\xi} \bfI_K]$, the solution of the latter SLE is simply obtained by dividing the last $K$ components of $\widehat{\bfz}$ by $\sqrt{\xi}$.
Denoting by $\bfz^t=[\bfu^t; \sqrt{\xi} \bfv^t]$ the estimate at iteration $t$ of KA, we can summarize this in Algorithm \ref{tab:Kac_MMBF}.

In terms of speed of convergence,  a direct calculation shows that the minimum gain of the matrix $\bfB^\herm$ over the subspace $\clX$ spanned by the columns of $\bfB$ is given by:
\begin{align}
\kappa_\clX (\bfB^\herm)&=\frac{ \lambda_{\min} (\bfB^\herm \bfB)}{\|\bfB\|_\sfF^2}=\frac{ \lambda_{\min} (\bfQ^\herm \bfQ) + \xi}{\|\bfQ\|_\sfF^2 + K \xi}.
\end{align} 
In particular,  when the channel vectors of the users have white complex Gaussian components (no spatial correlation), a direct calculation as in the case of MMSED shows that  the total computational complexity of KA scales as $O(\frac{KM}{(1-\sqrt{\rho})^2})$. To this we should also add $O(KM)$ operations needed for computing the sampling probability distribution and $O(KM)$ operation needed for calculating $\bfQ\bfw$ in Algorithm \ref{tab:Kac_MMBF}.

A direct inspection in Algorithm \ref{tab:Kac_MMSED} and \ref{tab:Kac_MMBF} also reveals that there is a symmetry between the MMSED for the UL and RZFBF for the DL: In the former, first the noisy received signal $\bfy$ is multiplied by $\bfQ^\herm$ and then the resulting vector $\bfb=\bfQ^\herm \bfy$ enters KA, which yields the estimate of the users symbols  $\widehat{\bfs}$, whereas in the latter, first the users symbols $\bfs$ to be transmitted in the DL enters KA and then the output of KA is multiplied by $\bfQ$ to yield the DL signal. Moreover, the speed of convergence is the same for both cases. This highly resembles the well-known 
UL-DL duality in the underlying vector Gaussian broadcast/multiple-access channel \cite{vishwanath2003duality}.


\section{Computation of the Detector/Precoder Matrix}\label{sec:matrix_design}
Up to now, we have used our proposed KA to compute an estimate of the users symbols from the received signal in the UL or to compute a suitable precoding vector in the DL. More specifically, we did not compute any UL/DL detection/precoding matrix from the available estimate of the channel state directly.  In this section, our goal is to show that our proposed KA can be generalized to compute the detection/precoding matrix as well. This results in a saving in computational complexity  when the channel coherence block 
is quite large such that the detection/precoding matrix corresponding to each coherence block is computed only once and  is then used for all the time-frequency signal dimensions inside the same coherence block. In the following, we focus on a specific coherence block and denote the estimate of the channel state of the users over  this coherence block by the $M\times K$ matrix $\bfQ$ as before. We consider the following procedure to compute the corresponding detection/precoding matrix.

In the DL, we use Algorithm \ref{tab:Kac_MMBF} with the input $\bfQ$ and run $K$ KAs in parallel where the input to the $i$-th KA is $\bfs_i=\bfe_i \in \bC^K$, where $\bfe_i$ denotes the $i$-th canonical basis, which has $1$ as its $i$-th component and is $0$ elsewhere. Note that the $K$ parallel KAs may share their randomness, i.e., they may use the same random row index  at each iteration,  or they may use their own independent randomness.   Let $\bfw_i \in \bC^K$ denote the output of the $i$-th KA after a suitable convergence.
Then, from the last step of Algorithm \ref{tab:Kac_MMBF}, namely, post-multiplication by $\bfQ$, the precoding vector $\bfy_i \in \bC^M$ in the DL produced by the $i$-th KA (corresponding to the symbol vector $\bfs_i=\bfe_i$) would be $\bfy_i=\bfQ \bfw_i$. However, due to the convergence of KA, $\bfy_i$ should be (approximately) equal  to the output of the RZF in the DL with the input $\bfs_i=\bfe_i$, i.e.,  
\begin{align}\label{dumm_bf_mat1}
\bfy_i=\bfQ(\bfQ^\herm \bfQ + \xi \bfI_K)^{-1} \bfs_i=\bfQ(\bfQ^\herm \bfQ + \xi \bfI_K)^{-1}\bfe_i,
\end{align}
which corresponds to the $i$-th column of the precoding matrix $\bfQ(\bfQ^\herm \bfQ + \xi \bfI_K)^{-1}$. 
Since $\bfQ$ is assumed to have a full column rank, from \eqref{dumm_bf_mat1} and the fact that $\bfy_i=\bfQ \bfw_i$, it immediately results that $\bfw_i$, after a suitable convergence, should correspond to the $i$-th column of $(\bfQ^\herm \bfQ + \xi \bfI_K)^{-1}$. 
Therefore, $(\bfQ^\herm \bfQ + \xi \bfI_K)^{-1}$ can be well approximated by the $K \times K$ matrix $\bfW=[\bfw_1, \dots, \bfw_K] \in \bC^{K \times K}$ consisting of the output produced by $K$ parallel KAs with corresponding inputs  $\{\bfs_i=\bfe_i: i\in [K]\}$. Also, the RZF matrix for the DL $\bfQ(\bfQ^\herm \bfQ + \xi \bfI_K)^{-1}$ can be approximated as $\bfG=\bfQ \bfW$. 
%
%
%
%
 Note that if the $K$ KAs can be run in parallel, computing the precoding matrix $\bfG$ will require $O(MK)$ operations including the $O(MK)$ computational complexity of finding the row selection probability (which is the same for all KAs and is shared among them).  Also, in terms of computational complexity, it is beneficial to keep the precoding matrix $\bfG$ in the factorized form $\bfG=\bfQ \bfW$  and compute the precoding vector as follows 
\begin{align}
\bfy=\bfG \bfs = \bfQ \bfW \bfs =(\bfQ (\bfW \bfs)),
\end{align}
where the parentheses in the last expression show the order of multiplication. Otherwise, we will need an additional $O(M K^2)$ operations to compute $\bfG$ directly, which increases the complexity by a factor of $O(K)$. 

Now consider the UL scenario. Since the detection matrix for the UL, $(\bfQ^\herm \bfQ + \xi \bfI_K)^{-1} \bfQ^\herm$ is the conjugate of the precoding matrix for the DL, we can obtain an approximation of the detection matrix as $\bfG^\herm=\bfW^\herm \bfQ^\herm$, where $\bfW$  denotes the $K \times K$ matrix produced at the output of the $K$ parallel KA in the DL scenario explained before\footnote{Note that, here, we decided to use Algorithm \ref{tab:Kac_MMBF} (initially designed for the DL) to compute beamforming matrices for both the UL and the DL. Instead, we could  apply Algorithm \ref{tab:Kac_MMSED} to $K$ parallel KAs with corresponding inputs $\{\bfb_i=\bfe_i: i \in [K]\}$ (see the definition of the variable $\bfb$ in Algorithm \ref{tab:Kac_MMSED}) to compute the beamforming matrix for the UL.}. 
%
%
%
Analogous to the DL, it is beneficial in terms of complexity to keep $\bfG^\herm$ in the factorized from and compute the estimate of users symbols from the received signal $\bfy \in \bC^M$ in the UL by
\begin{align}
\widehat{\bfs}=\bfG^\herm \bfy= \bfW^\herm \bfQ^\herm \bfy=(\bfW^\herm(\bfQ^\herm \bfy)),
\end{align}
where again the parentheses in the last expression show the order of multiplication. Recall that, assuming that  $K$ KAs  can be run in parallel, computing the detection matrix $\bfG^\herm$ (of course in its factorized form $\bfG^\herm=\bfW^\herm \bfQ^\herm$) requires $O(MK)$ operations. 
\section{Achievable Ergodic Rate}\label{sec:erg rate}

In this section, we derive lower and upper bounds on the achievable ergodic rate of our proposed KA.
We first prove an important (although an easy-to-see) property of KA.

\begin{proposition}\label{lin_prop}
	Consider an arbitrary SLE $\bfA \bfx =\bfb$ where $\bfA$ is the $m \times n$ matrix of equations, where $\bfx \in \bC^n$ is the set of unknowns, and where $\bfb \in \bC^m$ is the set of known coefficients. Let $\bfx^t$ be the estimate of  KA at iteration $t$ starting from the zero initialization $\bfx^0=0$. Then, $\bfx^t=\scrG^t(\bfA) \bfb$ where $\scrG^t(\bfA) \in \bC^{n\times m}$ is a linear operator that depends on $\bfA$ and the internal randomization of KA until the iteration $t$ but not on $\bfb$. \hfill $\square$
\end{proposition}

\begin{proof}
	We use induction on $t$. For $t=0$, we have $\bfx^0=0$ which is trivially a linear function (zero function) of $\bfb$. Let us assume that the induction hypothesis  is true for the $t$-th iteration, i.e., $\bfx^t=\scrG^t(\bfA) \bfb$, and prove the claim for $t+1$. Suppose $r(t) \in [m]$ is the index of the row selected at iteration $t$ and let us denote the conjugate of this row by $\bfa_{r(t)}$. Then, we have
	\begin{align}
	\bfx^{t+1}&=\bfx^t+\frac{ b_{r(t)}- \inp{\bfa_{r(t)}}{ \bfx^t} }{\| \bfa_{r(t)}\|^2}\bfa_{r(t)}
	\stackrel{(i)}{=} (\bfI_n - \frac{ \bfa_{r(t)} \bfa_{r(t)}^\herm }{\| \bfa_{r(t)}\|^2} ) \bfx^t+\frac{ \bfa_{r(t)} \bfe_{r(t)}^\herm}{\| \bfa_{r(t)}\|^2}\bfb\\
	&\stackrel{(ii)}{=}\left (  (\bfI_n- \frac{ \bfa_{r(t)} \bfa_{r(t)}^\herm }{\| \bfa_{r(t)}\|^2} ) \scrG^t(\bfA) +  \frac{ \bfa_{r(t)} \bfe_{r(t)}^\herm}{\| \bfa_{r(t)}\|^2} \right ) \bfb
	=:\scrG^{t+1}(\bfA) \bfb,\label{q_proof_2}
	\end{align}
	where in $(i)$ we defined $\bfe_{r(t)} \in \bC^m$ as the canonical vector that has $1$ at the index $r(t)$ and is zero elsewhere and replaced $b_{r(t)}=\bfe_{r(t)}^\herm \bfb$, and where in $(ii)$ we used the induction hypothesis that $\bfx^t=\scrG^t(\bfA) \bfb$ is a linear function of $\bfb$ with $\scrG^t(\bfA)$ not depending on $\bfb$. This implies that $\bfx^{t+1}=\scrG^{t+1}(\bfA) \bfb$, where it is seen that the linear operator $\scrG^{t+1}(\bfA) $ depends on $\bfA$ and on the internal randomization of KA until iteration $t+1$  but does not depend on $\bfb$. This proves the induction step and completes the proof.
\end{proof}

Using Proposition \ref{lin_prop}, we will derive lower and upper bounds on the achievable ergodic rate when the KA is run for a specific number of iterations $t$, which can  be generally much lower than $O(K)$ iterations  typically needed for its convergence (see, e.g., Section \ref{sec:scenarios}).
%
For simplicity, we focus on a UL scenario,  with a similar argument holding for a DL scenario. In the UL, the received signal at the BS is given by $\bfy=\bfH \bfs + \bfn$, where $\bfs \in \bC^K$ is the symbol sequence transmitted by the users, where $\bfH$ is the $M\times K$ matrix of the true channel state, and where $\bfn$ is the AWGN noise. 
From Proposition \ref{lin_prop}, it immediately results that,  assuming a zero initialization, the signal estimate $\widehat{\bfs}_t$ produced by KA at a specific iteration $t$ will be a random linear function of $\bfy$. More precisely, we have $\widehat{\bfs}_t =\scrG^t(\bfQ) \bfy$, where the $K\times M$ matrix $\scrG^t(\bfQ)$ depends on the noisy channel state $\bfQ$ and the internal randomization of KA but does not depend on  the received signal $\bfy$. 
 Letting  $\scrG^t(\bfQ)=[\bfg_1^t, \dots, \bfg_k^t]^\herm \in \bC^{K \times M}$ and considering the fact that the symbols estimate produced by KA is given by $\widehat{\bfs}_t =\scrG^t(\bfQ) \bfy$, we can interpret $\bfg_k^t \in \bC^M$, $k \in [K]$, as the detection vector corresponding to  the $k$-th user at $t$-th iteration. Note that the resulting detection matrix $\scrG^t(\bfQ)$ is generally far from MMSED matrix but converges to it when the KA is run for sufficiently many iterations $t$. Overall, one can write the time-varying (varying across coherence times due to channel randomness and varying across slots inside a coherence block due to internal randomness of KA) interference matrix as $\bfT^{\text{up}}= \scrG^t(\bfQ) \bfH\in \bC^{K \times K}$, where for simplicity we dropped the explicit dependence of $\bfT^{\text{up}}$ on the iteration $t$.

We use the upper and lower bounds on the achievable ergodic rate developed in \cite{fundamental_MM}. In particular, by treating the interference as noise and by coding across several coherence blocks, an upper bound on the ergodic capacity  is given by\footnote{This upper bound is obtained in \cite{caire2017ergodic} under the condition that the true channel state $\bfH$ is available at the BS such that the channel interference matrix $\bfT^{\text{up}}$ is fully known to the BS. This yields an upper bound on the achievable ergodic rate in our scenario, where only a noisy version of $\bfH$ is assumed to be known at the BS, thus, the channel interference matrix $\bfT^{\text{up}}= \scrG^t(\bfQ)^\herm \bfH$ is not perfectly known at the BS (since $\bfH$ is not known). }
\begin{align}
\overline{R}_k=\bE\left [ \log\left(1+ \frac{|\bfT^{\text{up}}_{k,k}|^2}{\sum_{k'\neq k} |\bfT^{\text{up}}_{k,k'}|^2 +\sigma^2_{k}}\right) \right ],\label{upper_bound_dumm}
\end{align}
where $\sigma^2_{k}=\frac{1}{\snr}\|\bfg_k^t\|^2$, with $\bfg_k^t$ denoting the conjugate of the $k$-th row of $\scrG^t(\bfQ)$, and where the expectation is taken over all the randomness of the channel state $\bfH$ and the detection matrix $\scrG^t(\bfQ)$ including the randomness due to KA. This yields the upper bound $\overline{R}=\sum_{k \in [K]} \overline{R}_k$ on the achievable sum rate. 
We will also use the following lower bound on the ergodic rate  from \cite{fundamental_MM}:
\begin{align}
\underline{R}_k=\log \left ( 1+ \frac{|\bE[\bfT^{\text{up}}_{k,k}]|^2}{\var(\bfT^{\text{up}}_{k,k})+ \sum_{k'\neq k} \bE[|\bfT^{\text{up}}_{k,k'}|^2]  +\bE[\sigma^2_{k}]}  \right),\label{lower_bound_dumm}
\end{align}
where $\var(.)$ denotes the variance of a random variable. This yields the lower bound $\underline{R}=\sum_{k \in [K]} \underline{R}_k$ on the ergodic rate. Following a similar argument, we obtain the corresponding upper/lower bounds on the ergodic capacity for the DL. 

In the next section, we use Monte-Carlo (MC) simulations  to empirically evaluate these lower and upper  bounds on the performance of our algorithm. 
%
Therefore, we need to be able to compute, at each instance of MC  simulation, the underlying matrix $\scrG^t(\bfQ)$ for a specific iteration $t$. Let us consider a specific instance of MC simulation and let us denote by $\clR_t:=\{r(l): l=1,2, \dots, t\}$ the random indices of the rows selected by KA up to iteration $t$ while decoding users symbols in a UL scenario. We can apply the iterative formula \eqref{q_proof_2} in the proof of Proposition \ref{lin_prop} 
to compute $\scrG^t(\bfQ)$ for each iteration $t$ explicitly. As an alternative, we can apply the technique proposed in Section \ref{sec:matrix_design} for computing the detector/precoder matrix. For example, in the UL scenario considered here, while running KA to decode users symbols from the noisy input in a specific instance of MC simulation, we also run in parallel $K$  KAs  with corresponding inputs $\{\bfs_i=\bfe_i: i \in [K]\}$ as explained in Section \ref{sec:matrix_design}. However, we impose the additional restriction that all the $K$ KAs use the same row indices as in the original KA used for decoding users symbols. 
In this way, we can compute,  at each instance  of MC simulation,  $\scrG^t(\bfQ)$ along with the estimate of users symbols.
Repeating the same procedure for each MC simulation, we  produce independent realizations of the random variables needed for computing upper/lower bounds in \eqref{upper_bound_dumm} and \eqref{lower_bound_dumm}, respectively. 
A similar technique can  be applied  to a DL scenario.

\section{Simulation results}\label{sec:results}
In this section, we evaluate the performance of our proposed technique empirically  via numerical simulations.
\subsection{Channel Model}
For all the simulations, we assume a correlated channel model for every user $k \in [K]$ where the true channel vector of the user $k$ is given by  $\bfh^k= \Phim^{1/2}\bfz^k$ where $\bfz^k \sim \cg\left( 0,\bfI_M\right)$ and where $\Phim$ is the covariance matrix of $\bfh^k$ with entries given by  $\left[ \bf \Phi\right] _{i,j}=a^{\left| i-j\right| }$ for some correlation parameter $a \in [0,1]$. We  assume that the entries of $\bfz^k$, thus, those of $\bfh^k$,  vary independently from one coherence period to another. We also assume that  the  estimated channel matrix $\bfQ$, available at the BS in each coherence period, is a noisy version of the true channel matrix $\bfH$ given by $\bfQ=\sqrt{1-\tau^2}\bfH+\tau\bfN$ where $\bfN$ is a matrix with i.i.d. entries  $\left[ \bfN\right] _{m,k}\sim \cg\left( 0,1\right) $, and where $\tau \in [0,1]$ is the estimation quality parameter with $\tau=0$ yielding a perfect channel estimation.

%

\subsection{Kaczmarz Residual Effect}
\label{sec:results:residual}

In this section, we investigate via numerical simulations the effect of the residual error in a KA directly applied to an OD SLE.          
We focus on a UL scenario where the noisy received signal at the BS is given by $\bfy=\bfH \bfs^\circ + \bfn^\circ$, where $\bfH$ is the true $M\times K$ channel matrix, where $\bfs^\circ \in \bC^K$ is the users symbols, and where $\bfn^\circ$ is the AWGN at BS antennas.  As explained in Sections \ref{sec:definition} and \ref{sec:scenarios}, neglecting the effect of the noise $\bfn^\circ$ and  applying the KA to the  OD SLE $\bfH {\bfs}=\bfy$ with the unknown ${\bfs}$ and the noisy input $\bfy$  results in a residual error such that the decoded users symbols do not converge to the estimate produced by the MMSE detector. 

This has been illustrated via numerical simulations in Fig.\,\ref{fig:compare_error_norm}, where it is seen that KA (Algorithm \ref{tab:Kac}) directly applied to OD SLE $\bfH \bfs= \bfy$ yields a considerably large estimation error, whereas our proposed KA yields an estimation error that approaches that of the MMSE detector after sufficiently many iterations. Note that for the simulations, we assume that the true channel state $\bfH$ is available at the BS to avoid the additional error caused by the noisy channel state and to pinpoint  only the residual error due to KA.
In Fig.\,\ref{fig:compare_error_norm}, we also compare the performance of our algorithm with that of another method  proposed in \cite{Herman}, which applies the KA to the UD SLE $[\bfH, \sqrt{\xi} \bfI_K] \bfz=\bfy$ with the noisy input $\bfy$, where $\xi>0$ is the noise parameter and  where $\bfz=[\bfs; \bfn] \in \bC^{M+K}$ contains both the symbol vector $\bfs\in \bC^K$ and the noise $\bfn\in \bC^M$. 
The underlying idea in \cite{Herman} is that the SLE $[\bfH, \sqrt{\xi} \bfI_K] \bfz=\bfy$ is always consistent, thus, the proposed KA avoids the effect of the residual (similarly to our proposed KA). However, a direct calculation based on our convergence analysis of KA in Section \ref{sec:math frame} reveals that the average gain of the resulting matrix $\bfB:=[\bfH, \sqrt{\xi} \bfI_K]$ in \cite{Herman} along the subspace $\clX$ produced by its rows (see, e.g., Definition \ref{avg_gain}) is given by 
$\kappa_\clX(\bfB)=\frac{\xi}{\|\bfH\|_\sfF^2 + K \xi}$, which is much lower than the gain $\frac{\lambda_{\min} (\bfH^\herm \bfH) + \xi}{\|\bfH\|_\sfF^2 + K \xi}$ corresponding to the matrix $[\bfH^\herm, \sqrt{\xi} \bfI_K]$ in our proposed KA (this is obtained from \eqref{MMSED_DL} assuming the perfect channel state $\bfQ=\bfH$). As a result, we expect that the KA in \cite{Herman} be extremely slower than our proposed KA, especially for high SNR ($\xi \to 0$). 
\begin{figure*}[t]
	\centering
	\input{compare_error_norm3.tex}
	\caption{\small Comparison of the detection error  in the UL, $\left\| \widehat{\bfs}-\bfs^\circ\right\| ^2$, versus number of iterations for different algorithms: original MMSED \eqref{MMSED_eq}, our proposed method (Algorithm \ref{tab:Kac_MMSED}), original KA (Algorithm \ref{tab:Kac}), and the KA proposed in \cite{Herman}. We consider $M=256$  antennas, $K=32$ users, and perfect channel state at the BS.}  
	\label{fig:compare_error_norm}
\end{figure*}
This is easily verified via numerical simulations, as illustrated in Fig.\,\ref{fig:compare_error_norm}. It is seen that at a relatively low SNR ($\text{SNR}=0$\,dB),  KA in \cite{Herman} ultimately convergess to the estimation error of the original MMSED but it converges much slower than our  proposed KA. At high SNR ($\text{SNR}=20$\,dB) where $\xi \to 0$, in contrast, we can identify two different convergence phases.  At first, KA in \cite{Herman} performs almost identically to the KA applied to the OD SLE $\bfH \bfs= \bfy$ (as expected due to $\xi \to 0$) and  reduces the estimation error rapidly  to the saturation limit of OD SLE $\bfH \bfs= \bfy$ caused by the noise residual effect. Then, the algorithm spends the rest of the iterations on reducing/compensating the residual error until it approaches the estimation error of the original MMSED, but this happens very slowly due to the small value of $\xi$.

\begin{figure*}[t]
	\centering
	\input{simple_Kac2.tex}
	\caption{\small Comparison of the upper/lower bounds (denoted by \lq ub\rq/\lq lb\rq \, respectively) on the per-user spectral efficiency in the UL versus SNR for different schemes: original MMSED \eqref{MMSED_eq}, our proposed method (Algorithm \ref{tab:Kac_MMSED}), original KA (Algorithm \ref{tab:Kac}), and the KA proposed in \cite{Herman}. We consider $M=256$ antennas, $K=32$ users, and  perfect channel state at the BS. The complexity budget is fixed to $40 MK$ for all of the schemes (except the original MMSED).
  }
	\label{fig:simple_Kac}
\end{figure*}

We also perform numerical simulations to investigate the effect of the residual error directly on the achievable ergodic rate in Massive MIMO. Fig.\,\ref{fig:simple_Kac} illustrates the numerical results. As expected, the residual error reduces the achievable rate considerably. 
It is also seen that although the MMSED in \cite{Herman} generally performs better than the simple KA (Algorithm \ref{tab:Kac}) at low SNRs, it completely loses its efficiency at high SNRs because of its extremely slow convergence. Our proposed method, however, is able to converge to the performance of the original MMSED in all  ranges of SNRs. 

\begin{figure*}
\input{down5.tex}
	\caption{\small Upper and lower bounds (denoted by \lq ub\rq/\lq lb\rq \, respectively) of the per-user ergodic spectral efficiency in the DL versus SNR for a loading factor of $\rho=\frac{K}{M}=\frac{1}{8}$, channel estimation quality factor $\tau=0.1$ and uncorrelated channels ($a=0$), are compared for the proposed method (denoted by \lq P\rq) against original linear precoders (ZFBF and RZFBF) and TPE scheme \cite{Bjornson}. The numbers at the end show computational complexity as a multiple of $MK$ which finally show the number of arithmetic operations. The computational complexity budget is assumed to be $10MK$ according to $J=3$ in TPE, $18MK$ according to $J=5$ in TPE, and $32MK$.}
	\label{fig:sim:down1}
\end{figure*}
\subsection{Downlink}
\label{sec:results:DL}
In the DL, we compare the performance of our method with that of the \textit{Truncated Polynomial Expansion} (TPE) technique proposed in \cite{Bjornson} in terms of spectral efficiency and total number of operations as a criterion for affordable computational complexity.
Fig.\,\ref{fig:sim:down1} illustrates  the upper  and the lower bound on the achievable per-user spectral efficiency versus SNR of our proposed algorithm.
It is seen from  Fig. \ref{fig:sim:down1} that for both ZFBF and RZFBF, $32MK$ operations guarantee a perfect match with the original ZFBF and RZFBF up to $20$\,dB SNR respectively, where for lower SNRs, a lower number of iterations is sufficient.  
It is seen that our proposed algorithm  has a slightly more complexity than the TPE proposed in \cite{Bjornson}. However, it has also the following crucial advantages compared with \cite{Bjornson}:
\begin{enumerate}
\item In TPE, matrix inversion is approximated via Taylor series expansion with a finite number of terms $J$, known as the TPE order. So for a certain value of $J$, one needs to compute $J$ coefficients of the expansion  in advance. For higher signal to noise ratios, typically higher value of $J$ is needed. This can be seen in Fig. \ref{fig:sim:down1} where for lower SNR values, lower $J$ can also make a good estimate. Moreover, when  $J$ changes, one needs to compute all the coefficients from scratch, which  incurs additional complexity. This is in contrast to our proposed method where only the number of iterations is needed to be increased or decreased for higher or lower values of SNR respectively. This can also be seen in Fig. \ref{fig:sim:down1}.

\item By changing the number of scheduled users, all the computations of the polynomial coefficients in TPE must be done from scratch. In our proposed method, however, everything is calculated each time a signal is going to be sent in the DL, or received in the UL, and hence no extra calculations are required.

\item When the channel response available at the BS changes quite fast with time (channel aging), or in OFDM-based systems where the channel response is available in only a subset of subcarriers (signal dimensions), the channel response needs to be, somehow, interpolated in time or in frequency respectively, to obtain the (approximate) channel response in other signal dimensions.
A scheme that computes directly the precoding vectors one-by-one, such as our proposed method, is much better suited to any channel tracking/prediction in time, interpolation in the OFDM subcarriers, etc. than a scheme such as TPE that computes the precoding/detection matrix from a subset of signal dimensions and then keeps it fixed for the remaining signal dimensions inside a coherence block. As a result, our proposed method can be combined with suitable channel prediction/interpolation techniques to obtain much better performances in situations where block fading channel model does not approximate the true channel very well. 

\item The proposed algorithm does not require the knowledge of channel estimation quality parameter $\tau$.

\end{enumerate}
It is also worthwhile to mention that, the number of operations  in both the proposed method and the TPE grows with the coherence period. For very large coherence periods traditional matrix-matrix multiplication will outperform both TPE and the proposed method, in terms of number of operations, although this happens only when the channel varies really slowly.

\begin{figure*}[t]	
\input{up8.tex}
	\caption{\small Upper and lower bounds (denoted by \lq ub\rq/\lq lb\rq \, respectively) of the per-user ergodic spectral efficiency in the UL versus SNR for a loading factor of $\rho=\frac{K}{M}=\frac{1}{8}$, channel estimation quality factor $\tau=0$ and correlation parameter $a=\left\lbrace 0,0.6\right\rbrace $ are compared for the proposed method (denoted by \lq P\rq) against original linear detectors (ZFD and MMSED), standard AMP \cite{Maleki_Studer_ISIT2015} and nonparametric AMP \cite{ghods2017optimally} (denoted by \lq S\rq and \lq N\rq\, respectively). The numbers at the end show computational complexity as a factor of $MK$ which finally show the number of arithmetic operations. The computational complexity budget is assumed to be $\left\lbrace 24MK,40MK\right\rbrace $.}	
	\label{fig:sim:up1}
\end{figure*}

\subsection{Uplink}
\label{sec:results:UL}
In the UL, we compare the performance of our algorithm with  a very recent \textit{approximate message passing} (AMP) technique proposed in \cite{ghods2017optimally, Maleki_Studer_ISIT2015}, especially with \cite{Maleki_Studer_ISIT2015} where the authors used AMP to build a \textit{nonparametric} MMSE detector for massive MIMO that is able to obtain an estimate of the signal and noise power. 
Fig. \ref{fig:sim:up1} illustrates our simulation results for the UL scenario in terms of upper  and lower bounds on the achievable per-user spectral efficiency. It is seen that in uncorrelated channels ($a=0$), for both ZFD and MMSED, $40MK$ operations are sufficient for our proposed method to achieve a perfect match to the upper and lower bounds of the original ZFD and MMSED up to $20$\,dB SNR respectively.\footnote{for the upper bounds, $32MK$ iterations suffice.} It is also seen that the AMP MMSED \cite{Maleki_Studer_ISIT2015} is able to converge to the original MMSED when the channel is uncorrelated ($a \approx 0$), with an order of $24MK$ operations.\footnote{This translates into $6$ iterations in AMP, since AMP carries out two matrix-vector operations in each iteration which will results in an order of $4MK$ additions and multiplications per iteration.} However, for a slightly correlated channel ($a\geq0.6$), neither the standard AMP \cite{Maleki_Studer_ISIT2015} nor the nonparametric AMP \cite{ghods2017optimally} are able to suitably recover the users symbols even for a large number of iterations whereas our proposed algorithm becomes slower due to channel correlation but still converges by slightly increasing the number of iterations. 
\subsection{Gap-to-Capacity versus Computational Complexity}
In this part, we compare our proposed method with the TPE \cite{Bjornson} and the AMP \cite{Maleki_Studer_ISIT2015,ghods2017optimally} in terms of   \textit{normalized gap-to-capacity} versus the number of operations $t$ given by $\frac{S_{\max}-S_t}{S_{\max}}$, where  $S_t$ denotes the upper/lower bound of per-user spectral efficiency of a specific beamforming/detection scheme after $t$ operations and where $S_{\max}$ denotes the   maximum possible upper/lower bound of per-user spectral efficiency achieved by original linear schemes. 
\begin{figure*}[]
\input{gap8.tex}
	\caption{\small Gap to the maximum spectral efficiency for the proposed scheme (denoted by \lq P\rq) is compared against that of the TPE \cite{Bjornson}, standard AMP \cite{Maleki_Studer_ISIT2015} and nonparametric AMP \cite{ghods2017optimally} (denoted by \lq S\rq and \lq N\rq\, respectively). 
		Loading factor, SNR, channel estimation quality parameter and channel correlation parameter are  $\rho=\frac{K}{M}=\frac{1}{8}$, $20$\,dB, $\tau=0$ and $a=\left\lbrace 0,.06\right\rbrace $ respectively.}	
	\label{fig:sim:gap}
\end{figure*}
This is illustrated in  Fig.\ \ref{fig:sim:gap} for different beamforming schemes in the DL and detection schemes in the UL for different channel correlations $a=0$ and $a=0.6$. It is seen that for uncorrelated channels, i.e $a=0$, the proposed method can achieve $\text{Gap}=10^{-2}$ with nearly $40MK$ operations. 
It is also worthwhile to mention that the AMP algorithm \cite{Maleki_Studer_ISIT2015,ghods2017optimally} is quite efficient for uncorrelated channels but fail in  correlated channels. 


%
%
\section{Conclusion}\label{sec:conclusion}

In this paper, we proposed a novel statistically robust approach for reducing  the computational complexity of beamforming in massive MIMO systems.  We showed that designing all variants of linear precoders/ detectors for the DL and the UL can be posed as the solution of an appropriate set of linear equations and applied a randomized variant of Kaczmarz algorithm to solve the resulting equations.  We also proposed a unified framework to analyze the performance of our proposed algorithm. Our algorithm does not require  knowledge of the statistics of users channel vectors and its convergence speed can be fully specified in terms of the condition number of the channel matrix, which depends on the number of BS antennas $M$, the number of served users $K$, as well as the statistical correlation of the channel vectors of the users. Our numerical simulations show that our algorithm is very competetive with the schemes proposed in the literature based on random matrix theory and approximate message passing in terms of computations and has a comparable performance. 
{\small
\bibliographystyle{IEEEtran}
\bibliography{references}
}

\end{document}

%% file: symbols.tex
\usepackage{mathbbol}

\def\mindex#1{\index{#1}}

\def\sq{\hbox{\rlap{$\sqcap$}$\sqcup$}}
\def\qed{\ifmmode\sq\else{\unskip\nobreak\hfil
\penalty50\hskip1em\null\nobreak\hfil\sq
\parfillskip=0pt\finalhyphendemerits=0\endgraf}\fi\medskip}

\long\def\defbox#1{\framebox[.9\hsize][c]{\parbox{.85\hsize}{%
\parindent=0pt
\baselineskip=12pt plus .1pt      
\parskip=6pt plus 1.5pt minus 1pt 
 #1}}}

\long\def\beginbox#1\endbox{\subsection*{}%
\hbox{\hspace{.05\hsize}\defbox{\medskip#1\bigskip}}%
\subsection*{}}

\def\endbox{}

\def\tr{\mathsf{tr}}

\newsavebox{\junk}
\savebox{\junk}[1.6mm]{\hbox{$|\!|\!|$}}

\def\argmin{\mathop{\rm arg\, min}}

\def\bC{{\mathbb C}}

\def\bE{{\mathbb E}}

\def\bP{{\mathbb P}}

\def\bR{{\mathbb R}}

\def\bV{{\mathbb V}}

\def\bfA{{\bf A}}
\def\bfB{{\bf B}}

\def\bfG{{\bf G}}
\def\bfH{{\bf H}}
\def\bfI{{\bf I}}

\def\bfN{{\bf N}}

\def\bfQ{{\bf Q}}

\def\bfT{{\bf T}}

\def\bfW{{\bf W}}
\def\bfX{{\bf X}}
\def\bfY{{\bf Y}}

\def\bfa{{\bf a}}
\def\bfb{{\bf b}}

\def\bfe{{\bf e}}

\def\bfg{{\bf g}}
\def\bfh{{\bf h}}

\def\bfn{{\bf n}}

\def\bfp{{\bf p}}
\def\bfq{{\bf q}}

\def\bfs{{\bf s}}

\def\bfu{{\bf u}}
\def\bfv{{\bf v}}
\def\bfw{{\bf w}}
\def\bfx{{\bf x}}
\def\bfy{{\bf y}}
\def\bfz{{\bf z}}

\def\scrG{{\mathscr{G}}}

\def\scrP{{\mathscr{P}}}

\def\sfF{{\sf F}}

\def\sfH{{\sf H}}

\def\sfa{{\sf a}}

\def\sfr{{\sf r}}

\def\bfmath#1{{\mathchoice{\mbox{\boldmath$#1$}}%
{\mbox{\boldmath$#1$}}%
{\mbox{\boldmath$\scriptstyle#1$}}%
{\mbox{\boldmath$\scriptscriptstyle#1$}}}}

\def\bfmY{\bfmath{Y}}

\def\bfmhhaY{\bfmath{\hhaY}} 
\def\bfmhhaY{\hbox to 0pt{$\widehat{\bfmY}$\hss}\widehat{\phantom{\raise 1.25pt\hbox{$\bfmY$}}}}

\def\til={{\widetilde =}}

\def\clC{{\cal C}}

\def\clN{{\cal N}}

\def\clR{{\cal R}}

\def\clX{{\cal X}}

 \def\FRAC#1#2#3{\genfrac{}{}{}{#1}{#2}{#3}}

\def\ddtp{{\mathchoice{\FRAC{1}{d^{\hbox to 2pt{\rm\tiny +\hss}}}{dt}}%
{\FRAC{1}{d^{\hbox to 2pt{\rm\tiny +\hss}}}{dt}}%
{\FRAC{3}{d^{\hbox to 2pt{\rm\tiny +\hss}}}{dt}}%
{\FRAC{3}{d^{\hbox to 2pt{\rm\tiny +\hss}}}{dt}}}}

\def\average#1,#2,{{1\over #2} \sum_{#1}^{#2}}

\def\eye(#1){{\bf(#1)}\quad}

\def\var{{\bV\sfa\sfr}}

\newtheorem{theorem}{{\bf Theorem}}
\newtheorem{proposition}{{\bf Proposition}}

\newtheorem{definition}{{\bf Definition}}

\newtheorem{remark}{{\bf Remark}}

\def\eq#1/{(\ref{e:#1})}

\newcommand{\inp}[2]{{\langle #1, #2 \rangle}}

\newcommand{\beqn}[1]{\notes{#1}%
\begin{eqnarray} \elabel{#1}}

\newcommand{\eeqn}{\end{eqnarray} }

\newcommand{\beq}[1]{\notes{#1}%
\begin{equation}\elabel{#1}}

\newcommand{\eeq}{\end{equation}}

\def\bdes{\begin{description}}
\def\edes{\end{description}}

\newcounter{rmnum}

\newcounter{anum}

{\end{list}}

\def\ass(#1:#2){(#1\ref{#1:#2})}

\def\ritem#1{
\item[{\sf \ass(\current_model:#1)}]
}

\newenvironment{recall-ass}[1]{%
\begin{description}
\def\current_model{#1}}{
\end{description}
}



%% file: macros.tex
\long\def\comment#1{}

\newcommand{\xv}{{\bf x}}

\newcommand{\Xm}{{\bf X}}

\newcommand{\Xc}{{\cal X}}

\newcommand{\Phim}{\hbox{\boldmath$\Phi$}}

\newcommand{\trace}{{\hbox{tr}}}

\newcommand{\SNR}{{\sf SNR}}

\newcommand{\transp}{{\sf T}}



%% file: optimal_p_plot2.tex
%
%

\begin{figure*}[t]
\def \mywidth {0.68}
\def \myshiftx {0.47}
\def \myshiftxx {0.94}
\begin{tikzpicture}
\pgfplotsset{every tick label/.append style={font=\small}}

\begin{scope}[scale=\mywidth]
\begin{axis}[%
xmin=0,
xmax=10,
ymin=0,
ymax=0.165,
ylabel={Eigen Values of $\overline{\scrP}$},
title={$M=256$ antennas, $K=10$ users},
title style={font=\small},
legend style={at={(1,1)}, anchor=north east, legend cell align=left, align=left, font=\small},
tick label style={/pgf/number format/fixed},
]
\addplot[ycomb, mark=o, line width=1pt, mark options={solid}] table[row sep=crcr] {%
1	0.100000000000065\\
2	0.100000000000043\\
3	0.100000000000036\\
4	0.100000000000018\\
5	0.100000000000003\\
6	0.0999999999999882\\
7	0.0999999999999712\\
8	0.0999999999999665\\
9	0.0999999999999503\\
10	0.0999999999999378\\
};

\addlegendentry{eigen values}

\addplot [dashed, color=red, line width=2pt]
table[row sep=crcr]{%
	1	0.0999999999999378\\
	10	0.0999999999999378\\
};
\addlegendentry{optimal $\kappa_\clX(\bfA, \bfp^*)$}

\addplot [line width=2 pt, dotted]
table[row sep=crcr]{%
1	0.1\\
10	0.1\\
};
\addlegendentry{upper bound $\frac{1}{K}$}

\addplot [color=blue, line width=1.5pt]
table[row sep=crcr]{%
1	0.076794301657377\\
10	0.076794301657377\\
};
\addlegendentry{suboptimal $\kappa_\clX(\bfA)$}

\end{axis}

\end{scope}

\begin{scope}[scale=\mywidth, xshift=\myshiftx \textwidth]
\begin{axis}[%
xmin=0,
xmax=15,
ymin=0,
ymax=0.11,
title={$M=256$ antennas, $K=15$ users},
title style={font=\small},
legend style={at={(1,1)}, anchor=north east, legend cell align=left, align=left, font=\small},
tick label style={/pgf/number format/fixed},
]
\addplot[ycomb, mark=o, line width=1pt, mark options={solid}] table[row sep=crcr] {%
1	0.0881117007687935\\
2	0.0853147246746934\\
3	0.0782959666671955\\
4	0.0730408619455419\\
5	0.0638515842338909\\
6	0.0628422639783677\\
7	0.0609492155530005\\
8	0.0609492116331862\\
9	0.0609492106945713\\
10	0.0609492104845745\\
11	0.0609492100808743\\
12	0.0609492100112839\\
13	0.0609492098761832\\
14	0.0609492098087568\\
15	0.0609492096527547\\
};

\addlegendentry{eigen values}

\addplot [dashed, color=red, line width=2pt]
table[row sep=crcr]{%
	1	0.0609492096527547\\
	15	0.0609492096527547\\
};
\addlegendentry{optimal $\kappa_\clX(\bfA, \bfp^*)$}

\addplot [line width=2 pt, dotted]
  table[row sep=crcr]{%
1	0.0666666666666667\\
15	0.0666666666666667\\
};
\addlegendentry{upper bound $\frac{1}{K}$}

\addplot [color=blue, line width=1.5pt]
  table[row sep=crcr]{%
1	0.0433450111505628\\
15	0.0433450111505628\\
};
\addlegendentry{suboptimal $\kappa_\clX(\bfA)$}

\end{axis}

\end{scope}

\begin{scope}[scale=\mywidth, xshift=\myshiftxx \textwidth]
\begin{axis}[%
xmin=0,
xmax=25,
ymin=0,
ymax=0.08,
title={$M=256$ antennas, $K=25$ users},
title style={font=\small},
legend style={at={(1,1)}, anchor=north east, legend cell align=left, align=left, font=\small},
tick label style={/pgf/number format/fixed},
]
\addplot[ycomb, mark=o, line width=1pt, mark options={solid}] table[row sep=crcr] {%
1	0.0699105656763463\\
2	0.0670705094126494\\
3	0.0609073744997444\\
4	0.0571788589853565\\
5	0.0534821473628355\\
6	0.049340016711708\\
7	0.0481361167342808\\
8	0.0451046914386902\\
9	0.0425297514042087\\
10	0.0408144680482412\\
11	0.0403138701740785\\
12	0.0373408332249105\\
13	0.0359269021485594\\
14	0.0342813765473377\\
15	0.0323790457746528\\
16	0.0305499928204447\\
17	0.0283037229360282\\
18	0.0283037196786197\\
19	0.0283037195887917\\
20	0.0283037195464017\\
21	0.0283037194930348\\
22	0.0283037194652518\\
23	0.0283037194555124\\
24	0.0283037194412078\\
25	0.0283037194311065\\
};

\addlegendentry{eigen values}

\addplot [dashed, color=red, line width=2pt]
table[row sep=crcr]{%
	1	0.0283037194311065\\
	25	0.0283037194311065\\
};
\addlegendentry{optimal $\kappa_\clX(\bfA, \bfp^*)$}

\addplot [line width=2 pt, dotted]
table[row sep=crcr]{%
1	0.04\\
25	0.04\\
};
\addlegendentry{upper bound $\frac{1}{K}$}

\addplot [color=blue, line width=1.5pt]
table[row sep=crcr]{%
1	0.0213267325207972\\
25	0.0213267325207972\\
};
\addlegendentry{suboptimal $\kappa_\clX(\bfA)$}

\end{axis}

\end{scope}

\end{tikzpicture}%

\caption{\small Comparison of the value $\kappa_{\clX}(\bfA, \bfp^*)$ for the optimal probability distribution $\bfp^*$ (obtained by solving the SDP) with that of the suboptimal one for a matrix $\bfA$ with i.i.d. $\cg(0,1)$ components and with $M=256$ rows (BS antennas) and $K \in \{10, 15, 25\}$ columns (users). The plots illustrate the eigen-value distribution of the average operator $\overline{\scrP}$ for the optimal distribution $\bfp^*$, where the minimum eigen value corresponds to $\kappa_{\clX}(\bfA,\bfp^*)$. The plots also illustrate the upper bound $\frac{1}{K}=\frac{1}{\min\{M,K\}}$ and the suboptimal $\kappa_\clX(\bfA)$ achieved with the suboptimal distribution $\bfp$ with $p_i=\frac{\|\bfa_i\|^2}{\|\bfA\|_\sfF^2}$. It is seen that for a matrix $\bfA$ with i.i.d. $\cg(0,1)$ components, the suboptimal distribution has also a good $\kappa_\clX(\bfA)$, which is quite close to the best $\kappa_{\clX}(\bfA, \bfp^*)$.}
\label{fig:opt_subopt}

\end{figure*}

\vspace{1cm}

%% file: simple_Kac2.tex
%
%
%
\begin{tikzpicture}

\begin{axis}[%
width=3in,
height=2in,
at={(5in,0.481in)},
scale only axis,
xmin=-20,
xmax=20,
xlabel style={font=\footnotesize},
xlabel={Transmit Power to Noise Ratio (SNR per receive antenna) [dB]},
ymin=0,
ymax=10,
ylabel style={font=\scriptsize},
ylabel={Average per-User Spectral Efficiency  [bit/s/Hz]},
axis background/.style={fill=white},
xmajorgrids,
ymajorgrids,
legend style={at={(0.01,0.99)}, anchor=north west, legend cell align=left, align=left, draw=white!15!black, font=\small}
]
\addplot [mark={triangle}, mark options={solid}, color=red , mark size=3pt, line width=1.2 pt]
table[row sep=crcr]{%
	-20	0.130640657155717\\
	-15	0.366087352214462\\
	-10	0.85671982443442\\
	-5	1.76554857348669\\
	0	3.03160912794368\\
	5	4.53918396478015\\
	10	6.16606427482049\\
	15	7.85810166133231\\
	20	9.45989815180692\\
};
\addlegendentry{original MMSED-ub}
\addplot [mark={square}, mark options={solid}, color=blue , mark size=3pt, line width=0.8 pt]
table[row sep=crcr]{%
	-20	0.109578284428963\\
	-15	0.320459877006773\\
	-10	0.811013947513637\\
	-5	1.72532480592047\\
	0	3.02814276362642\\
	5	4.53283943639187\\
	10	6.15144436209322\\
	15	7.78695503139128\\
	20	9.37867436353199\\
};
\addlegendentry{proposed MMSED-ub}
\addplot [mark=diamond, mark size=3pt, line width=1.2pt]
table[row sep=crcr]{%
-20	0.110450093558802\\
-15	0.317346680919515\\
-10	0.811906510962162\\
-5	1.7255538869069\\
0	3.0235173952725\\
5	4.43341114329832\\
10	5.17453299078138\\
15	5.74077747765152\\
20	6.88713667918871\\
};
\addlegendentry{MMSED in [44]-ub}

\addplot [mark=o, mark options={solid}, mark size=3pt, line width=1.2 pt]
 table[row sep=crcr]{%
	-20	0.018391832330035\\
	-15	0.0245087488554827\\
	-10	0.0644306717307928\\
	-5	0.266128655956252\\
	0	0.857192735134255\\
	5	1.84549204275539\\
	10	3.27935085608374\\
	15	4.93812610391084\\
	20	6.45850265793193\\
};
\addlegendentry{Alg.\,\ref{tab:Kac} MMSED-ub}
\addplot [mark={triangle}, dashed, mark options={solid}, color=black , mark size=2pt, line width=0.8 pt]
table[row sep=crcr]{%
	-20	0.130008166355991\\
	-15	0.314228805977818\\
	-10	0.796576465905935\\
	-5	1.79055750597775\\
	0	3.09626716746073\\
	5	4.52822598903809\\
	10	6.17389854204501\\
	15	7.74933127388974\\
	20	9.44285872719947\\
};
\addlegendentry{original MMSED-lb}
\addplot [dashed, mark={square}, mark options={solid}, color=blue , mark size=2pt, line width=1.2 pt]
table[row sep=crcr]{%
-20	0.14284397056804\\
-15	0.357468228798219\\
-10	0.828217346872508\\
-5	1.76442501828837\\
0	2.9878840548578\\
5	4.56257425762542\\
10	6.16874184656599\\
15	7.77139548679843\\
20	9.43865827741569\\
};
\addlegendentry{proposed MMSED-lb}
\addplot [dashed, mark=diamond, mark options=solid, line width=1.2 pt ]
table[row sep=crcr]{%
-20	0.121016614227966\\
-15	0.332894820359502\\
-10	0.8349281977955\\
-5	1.72674164353241\\
0	3.05236053392227\\
5	4.40441779355347\\
10	5.21910884019108\\
15	5.75015701751506\\
20	6.85811267993888\\
};
\addlegendentry{MMSED in [44]-lb}
\addplot [dashed, mark=o, mark options={solid}, mark size=2pt, line width=1.2 pt]
table[row sep=crcr]{%
	-20	0.0117043582293655\\
	-15	0.027318976557198\\
	-10	0.0702937507709507\\
	-5	0.285613713299698\\
	0	0.847568989444259\\
	5	1.91340703007941\\
	10	3.29813700748215\\
	15	4.86257605772515\\
	20	6.40608944246987\\
};
\addlegendentry{Alg.\,1 MMSED-lb}

\end{axis}
\end{tikzpicture}%

%% file: down5.tex
%
%
%

\begin{tikzpicture}
\def \myscale {0.71}
\def \myshiftx {0.46}
\def \myshiftxx {0.92}
\node (title)[font=\small] at ($(0.5 \textwidth,5.1)+(-.45,0)$) {Downlink, $\tau=0.1$, $M=256$, $K=32$};
\begin{scope}[scale=\myscale]
\begin{axis}[%
xmin=-20,
xmax=20,
xlabel style={font=\color{white!15!black}},
xlabel={SNR [dB]},
ymin=0,
ymax=9,
ylabel style={font=\color{white!15!black}},
ylabel={Average per-User Spectral Efficiency [bit/s/Hz]},
axis background/.style={fill=white},
title style={align=center},
title={i. $10MK$ operations },
xmajorgrids,
ymajorgrids,
legend style={at={(0,1)}, anchor=north west, legend cell align=left, align=left, draw=white!15!black, font=\small}
]
\addplot [mark={diamond}, mark options={solid}, color=red , mark size=3pt]
  table[row sep=crcr]{%
  	-20	0.0966737791013779\\
  	-15	0.285768717146925\\
  	-10	0.759005370535989\\
  	-5	1.67120012583167\\
  	0	2.97505106898094\\
  	5	4.4769846379007\\
  	10	6.00415398654976\\
  	15	7.39983962602803\\
  	20	8.46887083360748\\
  };
\addlegendentry{original ZFBF-ub}

\addplot [dashed, mark={diamond}, mark options={solid} ]
table[row sep=crcr]{%
	-20	0.109320367536398\\
	-15	0.315771961954787\\
	-10	0.80716748142875\\
	-5	1.71426959602235\\
	0	2.99792550816059\\
	5	4.48600232001149\\
	10	6.00740367500567\\
	15	7.40112843565225\\
	20	8.46951022645028\\
};
\addlegendentry{original RZFBF-ub}
\addplot [mark={triangle}, mark options={solid}, color={red}, mark size={3pt}]
table[row sep=crcr]{%
	-20	0.0962419841759125\\
	-15	0.284492343512309\\
	-10	0.761359586894504\\
	-5	1.67158424913767\\
	0	2.97989124946775\\
	5	4.48566974189232\\
	10	5.9855373832019\\
	15	7.37113900394701\\
	20	8.43550779160489\\
};
\addlegendentry{original ZFBF-lb}
\addplot [dashed, mark={triangle},mark options={solid}]
table[row sep=crcr]{%
	-20	0.108727190524971\\
	-15	0.314554713458063\\
	-10	0.808683135770147\\
	-5	1.71329446588181\\
	0	3.00203936596453\\
	5	4.49385291712365\\
	10	5.98859063591427\\
	15	7.37252998701895\\
	20	8.43657657054916\\
};
\addlegendentry{original RZFBF-lb}

\addplot [mark=o,mark size=3pt, mark options={solid}]
  table[row sep=crcr]{%
-20	0.10032584386544\\
-15	0.292476479086382\\
-10	0.758690154853733\\
-5	1.60730872908953\\
0	2.72942319700682\\
5	3.83902651553013\\
10	4.80995784382289\\
15	5.4758658530138\\
20	5.84262213560883\\
};
\addlegendentry{P ZFBF-ub-10}

\addplot [dashed, mark options={solid}, mark=o, color=blue]
  table[row sep=crcr]{%
-20	0.109038092027215\\
-15	0.312641256655318\\
-10	0.790170127175464\\
-5	1.63093849237539\\
0	2.72903459147711\\
5	3.86159907210686\\
10	4.80102377850433\\
15	5.45409537325157\\
20	5.86296069990682\\
};
\addlegendentry{P RZFBF-ub-10}

\addplot [mark=star, mark options={solid}]
table[row sep=crcr]{%
	-20	0.109451513269641\\
	-15	0.315514007546927\\
	-10	0.806873812319233\\
	-5	1.71327352597943\\
	0	2.98825147828615\\
	5	4.44603708473953\\
	10	5.87296761590331\\
	15	7.06440023094488\\
	20	7.8157874609492\\
};
\addlegendentry{TPE RZFBF-ub-10}
\addplot [mark={square}, mark size=3pt, mark options={solid}]
table[row sep=crcr]{%
	-20	0.0922548053699622\\
	-15	0.269251455606946\\
	-10	0.692839840040407\\
	-5	1.42961899462423\\
	0	2.21116030348117\\
	5	2.78562763370981\\
	10	3.08437552908241\\
	15	3.12514446395529\\
	20	3.13941161202662\\
};
\addlegendentry{P ZFBF-lb-10}

\addplot [dashed, mark={square}, mark options={solid}, color=blue]
table[row sep=crcr]{%
	-20	0.10052123975999\\
	-15	0.287162528986153\\
	-10	0.715290708406873\\
	-5	1.43179776572863\\
	0	2.2438393986446\\
	5	2.73146202937775\\
	10	3.11506881070736\\
	15	3.12447209907197\\
	20	3.14252513094665\\
};
\addlegendentry{P RZFBF-lb-10}

\addplot [dotted, mark=pentagon, mark options={solid}]
table[row sep=crcr]{%
	-20	0.108727190645453\\
	-15	0.314554724688715\\
	-10	0.80867047370231\\
	-5	1.71262346547329\\
	0	2.99257772045324\\
	5	4.45378725372581\\
	10	5.85176998542275\\
	15	7.0284740495276\\
	20	7.81100931486365\\
};
\addlegendentry{TPE RZFBF-lb-10}
\end{axis}
\end{scope}

\begin{scope}[scale=\myscale, xshift=\myshiftx \textwidth]
\begin{axis}[
xmin=-20,
xmax=20,
xlabel style={font=\color{white!15!black}},
xlabel={SNR [dB]},
title={ii. $18MK$ operations},
ymin=0,
ymax=9,
xmajorgrids,
ymajorgrids,
legend style={at={(0,1)}, anchor=north west, legend cell align=left, align=left, draw=white!15!black, font=\small}
]
\addplot [mark={diamond}, mark options={solid}, color=red , mark size=3pt]
table[row sep=crcr]{%
	-20	0.0966737791013779\\
	-15	0.285768717146925\\
	-10	0.759005370535989\\
	-5	1.67120012583167\\
	0	2.97505106898094\\
	5	4.4769846379007\\
	10	6.00415398654976\\
	15	7.39983962602803\\
	20	8.46887083360748\\
};
\addlegendentry{original ZFBF-ub}

\addplot [dashed, mark={diamond}, mark options={solid} ]
table[row sep=crcr]{%
	-20	0.109320367536398\\
	-15	0.315771961954787\\
	-10	0.80716748142875\\
	-5	1.71426959602235\\
	0	2.99792550816059\\
	5	4.48600232001149\\
	10	6.00740367500567\\
	15	7.40112843565225\\
	20	8.46951022645028\\
};
\addlegendentry{original RZFBF-ub}
\addplot [mark={triangle}, mark options={solid}, color={red}, mark size={3pt}]
table[row sep=crcr]{%
-20	0.0962419841759125\\
-15	0.284492343512309\\
-10	0.761359586894504\\
-5	1.67158424913767\\
0	2.97989124946775\\
5	4.48566974189232\\
10	5.9855373832019\\
15	7.37113900394701\\
20	8.43550779160489\\
};
\addlegendentry{original ZFBF-lb}
\addplot [dashed, mark={triangle},mark options={solid}]
table[row sep=crcr]{%
-20	0.108727190524971\\
-15	0.314554713458063\\
-10	0.808683135770147\\
-5	1.71329446588181\\
0	3.00203936596453\\
5	4.49385291712365\\
10	5.98859063591427\\
15	7.37252998701895\\
20	8.43657657054916\\
};
\addlegendentry{original RZFBF-lb}
\addplot [mark=o, mark size=3pt, mark options={solid}]
table[row sep=crcr]{%
	-20	0.0981644809038310\\
	-15	0.289493756629929\\
	-10	0.763536020975470\\
	-5	1.66542684310267\\
	0	2.92566380453055\\
	5	4.31213063993920\\
	10	5.62117906505725\\
	15	6.68899083300910\\
	20	7.41147431106220\\
};
\addlegendentry{P ZFBF-ub-18}
\addplot [dashed, mark=o , mark options={solid}, color=blue]
table[row sep=crcr]{%
	-20 0.109257679246159\\
	-15	0.315578836533865\\
	-10	0.803496474607672\\
	-5	1.70146070323734\\
	0	2.93266160736007\\
	5	4.32219181647844\\
	10	5.62972209233688\\
	15	6.67739359440322\\
	20	7.42661882513203\\
};
\addlegendentry{P RZFBF-ub-18}

\addplot [mark=star]
  table[row sep=crcr]{%
-20	0.109350942361082\\
-15	0.315838623830216\\
-10	0.807500949946666\\
-5	1.71435788258474\\
0	2.99791888123313\\
5	4.48511170634386\\
10	6.005295348606\\
15	7.39450674942096\\
20	8.45756880028631\\
};
\addlegendentry{TPE RZFBF-ub-18}
\addplot [mark={square}, mark size=3pt, mark options={solid}]
table[row sep=crcr]{%
	-20	0.0975736190730419\\
	-15	0.28581987316302\\
	-10	0.759168543105731\\
	-5	1.65553608967321\\
	0	2.79114471226115\\
	5	4.03291014714306\\
	10	5.17232129756361\\
	15	5.86169896644076\\
	20	6.40737093071547\\
};
\addlegendentry{P ZFBF-lb-18}

\addplot [dashed, mark={square}, mark options={solid},color=blue]
table[row sep=crcr]{%
	-20	0.108243491345739\\
	-15	0.311708610781446\\
	-10	0.795357991352816\\
	-5	1.65498103942845\\
	0	2.87905710198405\\
	5	4.10733248490041\\
	10	5.21839625033936\\
	15	5.89486133410311\\
	20	6.41159170496957\\
};
\addlegendentry{P RZFBF-lb-18}

\addplot [dotted, mark=pentagon, mark options={solid}]
table[row sep=crcr]{%
	-20	0.109211578981773\\
	-15	0.315819920094743\\
	-10	0.80577073065768\\
	-5	1.71618309130501\\
	0	2.99721412362788\\
	5	4.48967654720808\\
	10	5.99742046143355\\
	15	7.3799054506091\\
	20	8.42756589238901\\
};
\addlegendentry{TPE RZFBF-lb-18}
\end{axis}
\end{scope}

\begin{scope}[scale=\myscale, xshift=\myshiftxx \textwidth]
\begin{axis}[
xmin=-20,
xmax=20,
xlabel style={font=\color{white!15!black}},
xlabel={SNR [dB]},
title={iii. $32MK$ operations},
ymin=0,
ymax=9,
xmajorgrids,
ymajorgrids,
legend style={at={(0,1)}, anchor=north west, legend cell align=left, align=left, draw=white!15!black, font=\small}
]
\addplot [mark={diamond}, mark options={solid}, color=red , mark size=3pt]
table[row sep=crcr]{%
	-20	0.0966737791013779\\
	-15	0.285768717146925\\
	-10	0.759005370535989\\
	-5	1.67120012583167\\
	0	2.97505106898094\\
	5	4.4769846379007\\
	10	6.00415398654976\\
	15	7.39983962602803\\
	20	8.46887083360748\\
};
\addlegendentry{original ZFBF-ub}

\addplot [dashed, mark={diamond}, mark options={solid} ]
table[row sep=crcr]{%
	-20	0.109320367536398\\
	-15	0.315771961954787\\
	-10	0.80716748142875\\
	-5	1.71426959602235\\
	0	2.99792550816059\\
	5	4.48600232001149\\
	10	6.00740367500567\\
	15	7.40112843565225\\
	20	8.46951022645028\\
};
\addlegendentry{original RZFBF-ub}
\addplot [mark={triangle}, mark options={solid}, color={red}, mark size={3pt}]
table[row sep=crcr]{%
	-20	0.0962419841759125\\
	-15	0.284492343512309\\
	-10	0.761359586894504\\
	-5	1.67158424913767\\
	0	2.97989124946775\\
	5	4.48566974189232\\
	10	5.9855373832019\\
	15	7.37113900394701\\
	20	8.43550779160489\\
};
\addlegendentry{original ZFBF-lb}
\addplot [dashed, mark={triangle},mark options={solid}]
table[row sep=crcr]{%
	-20	0.108727190524971\\
	-15	0.314554713458063\\
	-10	0.808683135770147\\
	-5	1.71329446588181\\
	0	3.00203936596453\\
	5	4.49385291712365\\
	10	5.98859063591427\\
	15	7.37252998701895\\
	20	8.43657657054916\\
};
\addlegendentry{original RZFBF-lb}
\addplot [mark=o, mark size=3pt, mark options={solid}]
table[row sep=crcr]{
-20	0.0971082902808831\\
-15	0.286576507264834\\
-10	0.760665019379349\\
-5	1.67191544027174\\
0	2.97429116938925\\
5	4.46554791442641\\
10	5.96923499106237\\
15	7.32740787655281\\
20	8.31375350131984\\
};
\addlegendentry{P ZFBF-ub-32}
\addplot [dashed, mark=o, mark options=solid, color=blue]
  table[row sep=crcr]{%
-20	0.109348848004454\\
-15	0.315825019182977\\
-10	0.807399219312975\\
-5	1.71356226905513\\
0	2.99430188677071\\
5	4.47239130260917\\
10	5.97014677168548\\
15	7.32212221253282\\
20	8.33217038046601\\
};
\addlegendentry{P RZFBF-ub-32}

%
%
%
\addplot [mark=square, mark size=3pt, mark options={solid}]
table[row sep=crcr]{%
	-20	0.0961948287151554\\
	-15	0.285497090497416\\
	-10	0.763059377357684\\
	-5	1.67347632344007\\
	0	2.97691338812271\\
	5	4.46032808708203\\
	10	5.90335767262705\\
	15	7.18503474103531\\
	20	8.29979433945326\\
};
\addlegendentry{P ZFBF-lb-32}

\addplot [dashed, mark=square, mark options={solid}, color=blue]
table[row sep=crcr]{%
	-20	0.108726940190683\\
	-15	0.314546669618003\\
	-10	0.808609431604752\\
	-5	1.71195923509502\\
	0	2.99905075405401\\
	5	4.48350741166064\\
	10	5.94351015080977\\
	15	7.22217399489083\\
	20	8.30419433945326\\
};
\addlegendentry{P RZFBF-lb-32}

\end{axis}
\end{scope}
\end{tikzpicture}%

%% file: up8.tex
%
%
\def \mywidth {0.8}
%
\centering
\begin{tikzpicture}
\node (title)[font=\small] at ($(0.5 \textwidth,4.1)+(00.1,0.8)$) {Uplink, $\tau=0$, $M=256$, $K=32$};
\begin{scope}[scale=\mywidth, xshift= 0.7 \textwidth]
\begin{axis}[%
height=2.6in,
width=3.3in,
xmin=-20,
xmax=20,
xlabel style={font=\color{white!15!black}},
xlabel={SNR [dB]},
ymin=0,
ymax=10,
ylabel style={font=\color{white!15!black}},
axis background/.style={fill=white},
title style={align=center},
title={ii. a=0, 40MK operations},
xmajorgrids,
ymajorgrids,
legend style={at={(0,1)}, anchor=north west, legend cell align=left, align=left, draw=white!15!black, font=\footnotesize}
]

\addplot [mark={diamond}, color=red, mark size=3pt]
table[row sep=crcr]{%
	-20	0.098098283747272\\
	-15	0.288904389495747\\
	-10	0.766256137328556\\
	-5	1.68434153997691\\
	0	3.00676320546309\\
	5	4.53746226037159\\
	10	6.15047920555391\\
	15	7.80155546859727\\
	20	9.46829565193829\\
};
\addlegendentry{original ZFD-ub}

\addplot [dashed, mark options=solid, mark={diamond}]
table[row sep=crcr]{%
	-20	0.11021090501116\\
	-15	0.317414112311555\\
	-10	0.812807127633035\\
	-5	1.7253316367915\\
	0	3.02758947014415\\
	5	4.54556958751825\\
	10	6.15319351223836\\
	15	7.80244572079349\\
	20	9.46856265676501\\
};
\addlegendentry{original MMSED-ub}
\addplot [mark={triangle}, mark options={solid}, color={red}, mark size=3pt]
table[row sep=crcr]{%
	-20	0.111069992884109\\
	-15	0.28753453878206\\
	-10	0.796680418600827\\
	-5	1.67402626285586\\
	0	2.96963414399373\\
	5	4.54468289525765\\
	10	6.14322663273823\\
	15	7.81824327033833\\
	20	9.44000479222737\\
};
\addlegendentry{original ZFD-lb}

\addplot [dashed, mark={triangle},mark options={solid}]
table[row sep=crcr]{%
	-20	0.127638222968091\\
	-15	0.318278607620671\\
	-10	0.837011704392512\\
	-5	1.71839957238407\\
	0	2.98814495217816\\
	5	4.55554768298917\\
	10	6.14311158116778\\
	15	7.81933575294189\\
	20	9.44076741227711\\
};
\addlegendentry{original MMSED-lb}

\addplot [mark=o, mark size=3pt]
table[row sep=crcr]{%
	-20	0.0982268466785244\\
	-15	0.289234310524988\\
	-10	0.766909683429647\\
	-5	1.68526337140445\\
	0	3.00711155777889\\
	5	4.53748329836158\\
	10	6.13949268705118\\
	15	7.78395444702321\\
	20	9.32439376835045\\
};
\addlegendentry{P ZFD-ub-40}

\addplot [dashed, mark=o,mark options={solid}, color=blue]
table[row sep=crcr]{%
	-20	0.110210854424282\\
	-15	0.317414106578536\\
	-10	0.81280436936226\\
	-5	1.72530265537773\\
	0	3.02644996346735\\
	5	4.54121165273607\\
	10	6.14641389209573\\
	15	7.78715517355562\\
	20	9.39512217556055\\
};
\addlegendentry{P MMSED-ub-40}

\addplot [mark={square}, mark options={solid}, mark size=3pt]
table[row sep=crcr]{%
	-20	0.111372964234407\\
	-15	0.28739428236616\\
	-10	0.797318969438884\\
	-5	1.6746652638873\\
	0	2.97102959718221\\
	5	4.54665934795863\\
	10	6.1367532780046\\
	15	7.79742726834211\\
	20	9.35731337703254\\
};
\addlegendentry{P ZFD-lb-40}

\addplot [dashed, mark={square}, mark options={solid}, color=blue]
table[row sep=crcr]{%
	-20	0.126918820037751\\
	-15	0.318264561330655\\
	-10	0.836936146620195\\
	-5	1.7180523491154\\
	0	2.98318191747881\\
	5	4.55107453048947\\
	10	6.12686879431067\\
	15	7.74910052764583\\
	20	9.3072902628865\\
};
\addlegendentry{P MMSED-lb-40}
\end{axis}
\end{scope}
\begin{scope}[scale=\mywidth, xshift=0.2 \textwidth]
\begin{axis}[
height=2.6in,
width=3.3in,
xmin=-20,
xmax=20,
xlabel style={font=\color{white!15!black}},
xlabel={SNR [dB]},
ymin=0,
ymax=10,
xmajorgrids,
ymajorgrids,
title={i. a=0, 24MK operations},
ylabel style={font=\large},
legend style={at={(0,1)}, anchor=north west, legend cell align=left, align=left, draw=white!15!black, font=\footnotesize}
]
\addplot [mark={diamond}, color=red, mark size=3pt]
table[row sep=crcr]{%
	-20	0.098098283747272\\
	-15	0.288904389495747\\
	-10	0.766256137328556\\
	-5	1.68434153997691\\
	0	3.00676320546309\\
	5	4.53746226037159\\
	10	6.15047920555391\\
	15	7.80155546859727\\
	20	9.46829565193829\\
};
\addlegendentry{original ZFD-ub}

\addplot [dashed, mark options=solid, mark={diamond}]
table[row sep=crcr]{%
	-20	0.11021090501116\\
	-15	0.317414112311555\\
	-10	0.812807127633035\\
	-5	1.7253316367915\\
	0	3.02758947014415\\
	5	4.54556958751825\\
	10	6.15319351223836\\
	15	7.80244572079349\\
	20	9.46856265676501\\
};
\addlegendentry{original MMSED-ub}
\addplot [mark={triangle}, mark options={solid}, color={red}, mark size=3pt]
table[row sep=crcr]{%
	-20	0.111069992884109\\
	-15	0.28753453878206\\
	-10	0.796680418600827\\
	-5	1.67402626285586\\
	0	2.96963414399373\\
	5	4.54468289525765\\
	10	6.14322663273823\\
	15	7.81824327033833\\
	20	9.44000479222737\\
};
\addlegendentry{original ZFD-lb}

\addplot [dashed, mark={triangle},mark options={solid}]
table[row sep=crcr]{%
	-20	0.127638222968091\\
	-15	0.318278607620671\\
	-10	0.837011704392512\\
	-5	1.71839957238407\\
	0	2.98814495217816\\
	5	4.55554768298917\\
	10	6.14311158116778\\
	15	7.81933575294189\\
	20	9.44076741227711\\
};
\addlegendentry{original MMSED-lb}

\addplot [mark=o, mark size=3pt]
table[row sep=crcr]{%
	-20	0.0987429540396973\\
	-15	0.291429474697715\\
	-10	0.771072863335806\\
	-5	1.68718235374066\\
	0	2.99274157043969\\
	5	4.49119985089724\\
	10	5.97250556527328\\
	15	7.36741894842222\\
	20	8.59828050374192\\
};
\addlegendentry{P ZFD-ub-24}

\addplot [dashed, mark=o, mark options=solid, color=blue]
table[row sep=crcr]{%
	-20	0.110009029549701\\
	-15	0.318006236171303\\
	-10	0.813216841925411\\
	-5	1.72321803143558\\
	0	3.00983559433358\\
	5	4.48793514146069\\
	10	6.00968252599676\\
	15	7.39147815491353\\
	20	8.59765277407027\\
};
\addlegendentry{P MMSED-ub-24}

\addplot [mark=star]
table[row sep=crcr]{%
	-20	0.11021090501116\\
	-15	0.31741411231155\\
	-10	0.812807127461855\\
	-5	1.72533147963963\\
	0	3.02758093474751\\
	5	4.54550618108161\\
	10	6.15291931983726\\
	15	7.80099851627929\\
	20	9.46611980069475\\
};
\addlegendentry{S-AMP-MMSED-ub-24}

\addplot [mark=square,mark size=3pt]
table[row sep=crcr]{%
	-20	0.102171685948949\\
	-15	0.290550135671639\\
	-10	0.770635412522634\\
	-5	1.67998005318279\\
	0	2.95140740649812\\
	5	4.36628647271871\\
	10	5.72033138034241\\
	15	6.73678141852714\\
	20	7.20495906750073\\
};
\addlegendentry{P ZFD-lb-24}

\addplot [dashed, mark=square, mark options={solid},color=blue]
table[row sep=crcr]{%
	-20	0.113301862754382\\
	-15	0.317048635462504\\
	-10	0.813766936335789\\
	-5	1.71324449585942\\
	0	2.96669051143985\\
	5	4.36259239503133\\
	10	5.67604391906277\\
	15	6.80295629212227\\
	20	7.25647414076634\\
};
\addlegendentry{P MMSED-lb-24}

\addplot [dotted, mark=pentagon, mark options={solid}]
table[row sep=crcr]{%
	-20	0.127638222978137\\
	-15	0.318278606038076\\
	-10	0.837011155186259\\
	-5	1.71839898385307\\
	0	2.98807292704208\\
	5	4.55595348267342\\
	10	6.14262731015732\\
	15	7.819995521469\\
	20	9.43644551128068\\
};
\addlegendentry{S-AMP MMSED-lb-24}
\end{axis}
\end{scope}

\begin{scope}[scale=\mywidth, xshift=0.7 \textwidth, yshift=-0.43 \textwidth]
\begin{axis}[
height=2.6in,
width=3.3in,
xmin=-20,
xmax=20,
xlabel style={font=\color{white!15!black}},
xlabel={SNR [dB]},
ymin=0,
ymax=10,
title={iv. a=0.6, 40MK operations},
xmajorgrids,
ymajorgrids,
legend style={at={(0,1)}, anchor=north west, legend cell align=left, align=left, draw=white!15!black, font=\footnotesize}
]
\addplot [mark=diamond,color=red, mark size=3pt]
table[row sep=crcr]{%
	-20	0.0861715196339094\\
	-15	0.256096024374698\\
	-10	0.691232598405306\\
	-5	1.55463477406733\\
	0	2.82892607257811\\
	5	4.35438145110961\\
	10	5.95841747005181\\
	15	7.60587056547712\\
	20	9.25485149991588\\
};
\addlegendentry{original ZFD-ub}

\addplot [dashed,mark=diamond, mark options={solid}]
table[row sep=crcr]{%
	-20	0.108974670669689\\
	-15	0.310602614345994\\
	-10	0.780693464224801\\
	-5	1.6375707775121\\
	0	2.87448091158157\\
	5	4.37216501352722\\
	10	5.9645092067329\\
	15	7.60785089039974\\
	20	9.25548796947762\\
};
\addlegendentry{original MMSED-ub}


\addplot [mark=triangle, mark options={solid}, color=red, mark size=3pt]
table[row sep=crcr]{%
	-20	0.0857963977734561\\
	-15	0.279176409751418\\
	-10	0.707431048720751\\
	-5	1.62162731393661\\
	0	2.82717963049401\\
	5	4.39047113993479\\
	10	6.00178484296055\\
	15	7.57217576628723\\
	20	9.27516807190871\\
};
\addlegendentry{original ZFD-lb}

\addplot [dashed, mark=triangle, mark options={solid}]
table[row sep=crcr]{%
	-20	0.106016884513526\\
	-15	0.329414425372057\\
	-10	0.788822435765422\\
	-5	1.6976561092053\\
	0	2.88472901827461\\
	5	4.41618049133629\\
	10	6.0113599435257\\
	15	7.5748576117616\\
	20	9.27654292605063\\
};
\addlegendentry{original MMSED-lb}
\addplot [mark=o, mark size=3pt]
table[row sep=crcr]{%
	-20	0.0868780244399208\\
	-15	0.258351517280486\\
	-10	0.693255966554687\\
	-5	1.55968446104269\\
	0	2.83527281949663\\
	5	4.34038428886176\\
	10	5.92639116703161\\
	15	7.49624371321558\\
	20	8.9411415472445\\
};
\addlegendentry{P ZFD-ub-40}
\addplot [dashed, mark=o, mark options=solid,color=blue]
table[row sep=crcr]{%
	-20	0.109149175833333\\
	-15	0.311367057157945\\
	-10	0.778806081043397\\
	-5	1.63705800334225\\
	0	2.87455569360723\\
	5	4.34991284876752\\
	10	5.92876493169673\\
	15	7.49471467729036\\
	20	8.92375816169972\\
};
\addlegendentry{P MMSED-ub-40}

\addplot [mark=star, mark size=3pt]
table[row sep=crcr]{%
	-20	0.109149184202277\\
	-15	0.311368345300269\\
	-10	0.77882023886121\\
	-5	1.63446951462333\\
	0	2.77900200732156\\
	5	3.9163856419768\\
	10	4.7222575222899\\
	15	5.42337503485911\\
	20	5.69171672595967\\
};
\addlegendentry{S-AMP-MMSED-ub-40}
\addplot [mark=square, mark options={solid},mark size=3pt]
table[row sep=crcr]{%
	-20	0.0866115104635836\\
	-15	0.280443282690849\\
	-10	0.711361206405101\\
	-5	1.627320384831\\
	0	2.83258489163944\\
	5	4.38241404152848\\
	10	5.94312236930565\\
	15	7.26855242019608\\
	20	8.8213600236142\\
};
\addlegendentry{P ZFD-lb-40}

\addplot [dashed, mark=square, mark options={solid},color=blue]
table[row sep=crcr]{%
	-20	0.106011578110602\\
	-15	0.329413923559625\\
	-10	0.788609071767215\\
	-5	1.69700786103146\\
	0	2.88293336937049\\
	5	4.40901488064811\\
	10	5.91596957738478\\
	15	7.44027344657257\\
	20	8.80279228786978\\
};
\addlegendentry{P MMSED-lb-40}

\addplot [mark=halfsquare right*, mark size=3pt]
table[row sep=crcr]{%
	-20	0.0218842776392572\\
	-15	0.189168392147228\\
	-10	0.709926050281639\\
	-5	1.61325468128587\\
	0	2.73731047410299\\
	5	3.59176901750064\\
	10	3.61098492988312\\
	15	3.50892372821689\\
	20	3.43449800751768\\
};
\addlegendentry{N-AMP MMSED-lb-40}
\end{axis}
\end{scope}

\begin{scope}[scale=\mywidth, yshift=-0.43 \textwidth, xshift=0.2 \textwidth]
\begin{axis}[
height=2.6in,
width=3.3in,
xmin=-20,
xmax=20,
xlabel style={font=\color{white!15!black}},
xlabel={SNR [dB]},
title={iii. a=0.6, 24MK operations},
ymin=0,
ymax=10,
xmajorgrids,
ymajorgrids,
ylabel={Average per-User Spectral Efficiency [bit/s/Hz]},
ylabel style={at={(-0.2,1.2)}, font=\large},
legend style={at={(0,1)}, anchor=north west, legend cell align=left, align=left, draw=white!15!black, font=\footnotesize}
]
\addplot [mark=diamond,color=red, mark size=3pt]
table[row sep=crcr]{%
	-20	0.0861715196339094\\
	-15	0.256096024374698\\
	-10	0.691232598405306\\
	-5	1.55463477406733\\
	0	2.82892607257811\\
	5	4.35438145110961\\
	10	5.95841747005181\\
	15	7.60587056547712\\
	20	9.25485149991588\\
};
\addlegendentry{original ZFD-ub}

\addplot [dashed,mark=diamond, mark options={solid}]
table[row sep=crcr]{%
	-20	0.108974670669689\\
	-15	0.310602614345994\\
	-10	0.780693464224801\\
	-5	1.6375707775121\\
	0	2.87448091158157\\
	5	4.37216501352722\\
	10	5.9645092067329\\
	15	7.60785089039974\\
	20	9.25548796947762\\
};
\addlegendentry{original MMSED-ub}


\addplot [mark=triangle, mark options={solid}, color=red, mark size=3pt]
table[row sep=crcr]{%
	-20	0.0857963977734561\\
	-15	0.279176409751418\\
	-10	0.707431048720751\\
	-5	1.62162731393661\\
	0	2.82717963049401\\
	5	4.39047113993479\\
	10	6.00178484296055\\
	15	7.57217576628723\\
	20	9.27516807190871\\
};
\addlegendentry{original ZFD-lb}

\addplot [dashed, mark=triangle, mark options={solid}]
table[row sep=crcr]{%
	-20	0.106016884513526\\
	-15	0.329414425372057\\
	-10	0.788822435765422\\
	-5	1.6976561092053\\
	0	2.88472901827461\\
	5	4.41618049133629\\
	10	6.0113599435257\\
	15	7.5748576117616\\
	20	9.27654292605063\\
};
\addlegendentry{original MMSED-lb}
\addplot [mark=o, mark size=3pt]
table[row sep=crcr]{%
	-20	0.0883088533456923\\
	-15	0.262102656961979\\
	-10	0.70321251553223\\
	-5	1.56606507252738\\
	0	2.81223839783521\\
	5	4.22931329647162\\
	10	5.57441309864862\\
	15	6.68688544119785\\
	20	7.41996712514464\\
};
\addlegendentry{P ZFD-ub-24}
\addplot [dashed, mark=o, mark options={solid},color=blue]
table[row sep=crcr]{%
	-20	0.108972828458263\\
	-15	0.31056585072498\\
	-10	0.780069251428157\\
	-5	1.63156378112359\\
	0	2.84130545579392\\
	5	4.24663217102379\\
	10	5.59835457646432\\
	15	6.66885533603229\\
	20	7.46443035943383\\
};
\addlegendentry{P MMSED-ub-24}

\addplot [mark=star, mark size=3pt]
table[row sep=crcr]{%
	-20	0.108974670669689\\
	-15	0.310602614229201\\
	-10	0.780687668211375\\
	-5	1.63515554962935\\
	0	2.79398321385902\\
	5	3.82170968218442\\
	10	4.81135367354413\\
	15	5.3705904736752\\
	20	5.5297299930551\\
};
\addlegendentry{S-AMP-MMSED-ub-24}

\addplot [dashed, mark=square, mark options={solid}, mark size=3pt]
table[row sep=crcr]{%
	-20	0.0852570067127058\\
	-15	0.260998455212775\\
	-10	0.695785222034416\\
	-5	1.5461450764511\\
	0	2.76456530116299\\
	5	4.06592814944918\\
	10	5.21010848965163\\
	15	5.83396805019437\\
	20	6.30175428984062\\
};
\addlegendentry{P ZFD-lb-24}

\addplot [dashed, mark=square, mark options={solid},color=blue]
table[row sep=crcr]{%
	-20	0.107869914064152\\
	-15	0.314002877322073\\
	-10	0.772971454172973\\
	-5	1.61320052323128\\
	0	2.8024290220974\\
	5	4.11165763646618\\
	10	5.25119913788013\\
	15	5.85581624769175\\
	20	6.31723311344545\\
};
\addlegendentry{P MMSED-lb-24}
\addplot [mark=pentagon, mark size=3pt]
table[row sep=crcr]{%
-20	0.112812198200514\\
-15	0.309418954655284\\
-10	0.772329601171939\\
-5	1.62361724728489\\
0	2.78391446893415\\
5	3.67053662927329\\
10	3.96727381327727\\
15	3.94474234317207\\
20	4.14084196626952\\
};
\addlegendentry{S-AMP-MMSED-lb-24}
\end{axis}
\end{scope}
\end{tikzpicture}%

%% file: gap8.tex
%
%
\def \mywidth {0.8}
\definecolor{mycolor1}{rgb}{0.00000,0.44700,0.74100}%
\definecolor{mycolor2}{rgb}{0.85000,0.32500,0.09800}%
\definecolor{mycolor3}{rgb}{0.92900,0.69400,0.12500}%
\centering
\begin{tikzpicture}
\node (title)[font=\small] at ($(0.5 \textwidth,4.1)+(0,0.7)$) {Gap, $\tau=0$, $M=256$, $K=32$, $\text{SNR}=20$\,dB};
\pgfplotsset{every tick label/.append style={font=\small}}
\begin{scope}[scale=\mywidth, xshift=0.2 \textwidth]
\begin{axis}[%
height=2.5in,
width=3.2in,
xmin=16,
xmax=336,
xlabel style={font=\color{white!15!black}},
xlabel={Complexity Budget},
ymode=log,
ymin=0.01,
ymax=30,
yminorticks=true,
ylabel style={font=\color{white!15!black}},
ylabel={Gap to Expected Spectral Efficiency (normalized)},
ylabel style={at={(-0.2,-0.2)},font=\large},
axis background/.style={fill=white},
title style={align=center},
title={i. Uplink, $a=0$},
xmajorgrids,
ymajorgrids,
legend style={legend cell align=left, align=left, font=\small, draw=white!15!black},
xtick={16,64,128,192,256,320},
xticklabels={$2MK$, $8MK$, $16MK$, $24MK$, $32MK$, $40MK$ },
]

\addplot [mark=o, mark options={solid}, mark size=3pt]
table[row sep=crcr]{%
	64	0.428730615356514\\
	128	0.229494888390253\\
	192	0.0899355786601721\\
	256	0.0285709948446521\\
	320	0.0070829375383263\\
};
\addlegendentry{P ZFD-ub}

\addplot [dashed, mark=o, mark options={solid}, color=blue]
table[row sep=crcr]{%
64	0.430209674735074\\
128	0.227795875120953\\
192	0.0902549286055839\\
256	0.0285491901206969\\
320	0.0077734124115059\\
};
\addlegendentry{P MMSED-ub}

Mercedes star
\addplot [mark=square, mark size=3pt]
table[row sep=crcr]{%
64	0.351700732549757\\
128	0.0205819441460837\\
192	0.000371736926556566\\
256	7.29984295295347e-06\\
320	1.85057012915073e-07\\
};
\addlegendentry{S-AMP-MMSED-ub}
\addplot [mark=triangle, mark size=3pt]
table[row sep=crcr]{%
	64	0.729491224248499\\
	128	0.469396406164959\\
	192	0.238399290661655\\
	256	0.0774674250555506\\
	320	0.0177871717082471\\
};
\addlegendentry{P ZFD-lb}

\addplot [dashed, mark=triangle, mark options=solid, color=blue]
table[row sep=crcr]{%
	64	0.73034919431798\\
	128	0.470466579586623\\
	192	0.238789311619535\\
	256	0.0785699139339221\\
	320	0.0201706370514779\\
};
\addlegendentry{P MMSED-lb}

\addplot [dashed, mark options=solid, mark=square]
table[row sep=crcr]{%
	64	0.358320075718639\\
	128	0.0214707365255854\\
	192	0.000398012898654159\\
	256	8.28413780080527e-06\\
	320	6.02372696975326e-07\\
};
\addlegendentry{S-AMP-MMSED-lb}
\addplot [dashdotted, mark=star, mark options={solid}, color=blue]
table[row sep=crcr]{%
	64	0.362020292212401\\
	128	0.0219866982137076\\
	192	0.000424210885888016\\
	256	2.14288823477704e-05\\
	320	1.38216066653901e-05\\
};
\addlegendentry{N-AMP-MMSED-lb}
\end{axis}
\end{scope}

\begin{scope}[scale=\mywidth, xshift= 0.7 \textwidth]
\begin{axis}[
height=2.5in,
width=3.2in,
xmin=16,
xmax=336,
ymax=30,
xlabel style={font=\color{white!15!black}},
xlabel={Complexity Budget},
title={iii. Downlink, $a=0$},
ymode=log,
ymin=0.01,
xmajorgrids,
ymajorgrids,
xtick={16,64,128,192,256,320},
xticklabels={$2MK$, $8MK$, $16MK$, $24MK$, $32MK$, $40MK$ },
legend style={legend cell align=left, align=left, font=\small, draw=white!15!black},
]
\addplot [mark=o, mark size=3pt]
table[row sep=crcr]{%
	16	0.725963582284638\\
	80	0.351954539552071\\
	144	0.157186271413483\\
	208	0.0557743575364888\\
	272	0.0173131084409029\\
	336	0.00431357959702135\\
};
\addlegendentry{P ZFBF-ub}

\addplot [dashed, mark=o, mark options=solid, color=blue]
table[row sep=crcr]{%
	16	0.726966139915754\\
	80	0.350944719462124\\
	144	0.15695731203595\\
	208	0.056076322526738\\
	272	0.017241565545806\\
	336	0.00444915095534623\\
};
\addlegendentry{P RZFBF-ub}

\addplot [mark=diamond, mark size=3pt]
table[row sep=crcr]{%
	16	0.659712963308305\\
	80	0.115458490852938\\
	144	0.00281573140783274\\
	208	5.42703544006911e-05\\
};
\addlegendentry{TPE RZFBF-ub}
\addplot [mark=triangle, mark size=3pt]
table[row sep=crcr]{%
	16	0.906670770740946\\
	80	0.565189497247524\\
	144	0.273128325449334\\
	208	0.0916907745435216\\
	272	0.0307707861864738\\
	336	0.00448393342927144\\
};
\addlegendentry{P ZFBF-lb}

\addplot [dashed, mark=triangle, mark options=solid, color=blue]
table[row sep=crcr]{%
	16	0.906304568072334\\
	80	0.554191845089504\\
	144	0.288660977283373\\
	208	0.0794001482534594\\
	272	0.0280346375025186\\
	336	0.00632864982919271\\
};
\addlegendentry{P RZFBF-lb}

\addplot [dashed, mark options=solid, mark=diamond]
table[row sep=crcr]{%
	16	0.563662737401477\\
	80	0.0688596408645129\\
	144	0.00142075630501654\\
	208	1.24862346631629e-005\\
	272	0.00542273973157869\\
};
\addlegendentry{TPE RZFBF-lb}
\end{axis}
\end{scope}

\begin{scope}[scale=\mywidth, xshift= 0.2 \textwidth, yshift= -0.41 \textwidth]
\begin{axis}[
legend columns=2,
height=2.5in,
width=3.2in,
xmin=16,
xmax=336,
ymax=30,
xlabel style={font=\color{white!15!black}},
xlabel={Complexity Budget},
ymode=log,
ymin=0.01,
xmajorgrids,
ymajorgrids,
xtick={16,64,128,192,256,320},
title={ii. Uplink, $a=0.6$},
xticklabels={$2MK$, $8MK$, $16MK$, $24MK$, $32MK$, $40MK$ },
legend style={at={(-.01,.97)}, anchor=north west, font=\footnotesize, legend cell align=left, align=left, draw=white!15!black},
]
\addplot [mark=o, mark size=3pt]
table[row sep=crcr]{%
	64	0.546474935795056\\
	128	0.352377720204628\\
	192	0.196164522344946\\
	256	0.0844203333616356\\
	320	0.0326005294546252\\
};
\addlegendentry{P ZFD-ub}

\addplot [dashed, mark=o, mark options=solid, color=blue]
table[row sep=crcr]{%
64	0.546844079657982\\
128	0.354528056219072\\
192	0.191626764390432\\
256	0.0875678479732134\\
320	0.0309151180651839\\
};
\addlegendentry{P MMSED-ub}

\addplot [mark=square, mark size=3pt]
table[row sep=crcr]{%
64	0.581511453060513\\
128	0.422196091035524\\
192	0.377830532617281\\
256	0.365158900683391\\
320	0.354018261613846\\
};
\addlegendentry{S-AMP-MMSED-ub}
\addplot [mark=triangle, mark size=3pt]
table[row sep=crcr]{%
	64	0.75789675501496\\
	128	0.529949939174879\\
	192	0.321377355442987\\
	256	0.151758489822113\\
	320	0.0576800457515554\\
};
\addlegendentry{P ZFD-lb}

\addplot [dashed, mark=triangle, mark options=solid, color=blue]
table[row sep=crcr]{%
	64	0.757822713577891\\
	128	0.52801230401288\\
	192	0.321829708773218\\
	256	0.158282711774686\\
	320	0.0529569891988974\\
};
\addlegendentry{P MMSED-lb}

\addplot [dashed, mark=square, mark options=solid]
table[row sep=crcr]{%
	64	0.597364478527097\\
	128	0.508887455844496\\
	192	0.531921959483672\\
	256	0.621744537585464\\
	320	0.684588894193145\\
};
\addlegendentry{S-AMP-MMSED-lb}

\addplot [dashdotted, mark=star, mark options=solid, color=blue]
table[row sep=crcr]{%
	64	0.608881635731455\\
	128	0.512928584696516\\
	192	0.523131507087359\\
	256	0.557284468377805\\
	320	0.594972294150478\\
};
\addlegendentry{N-AMP-MMSED-lb}
\end{axis}
\end{scope}

\begin{scope}[scale=\mywidth, xshift=0.7 \textwidth, yshift= -0.41 \textwidth]
\begin{axis}[
height=2.5in,
width=3.2in,
xmin=16,
xmax=336,
ymax=30,
xlabel style={font=\color{white!15!black}},
xlabel={Complexity Budget},
ymode=log,
ymin=0.01,
xmajorgrids,
ymajorgrids,
xtick={16,64,128,192,256,320},
title={iv. Downlink, $a=0.6$},
xticklabels={$2MK$, $8MK$, $16MK$, $24MK$, $32MK$, $40MK$ },
legend style={legend cell align=left, align=left, font=\small, draw=white!15!black},
]
\addplot [mark=o, mark size=3pt]
table[row sep=crcr]{%
	16	0.761456036016903\\
	80	0.446754516266363\\
	144	0.252068581062475\\
	208	0.125256445873111\\
	272	0.0534314352786215\\
	336	0.0207729445270967\\
};
\addlegendentry{P ZFBF-ub}

\addplot [dashed, mark=o, mark options=solid, color=blue]
table[row sep=crcr]{%
	16	0.762340389272876\\
	80	0.448431185116021\\
	144	0.252839635067928\\
	208	0.125631906746846\\
	272	0.0546932854709683\\
	336	0.0206432473177367\\
};
\addlegendentry{P RZFBF-ub}

\addplot [mark=diamond, mark size=3pt]
table[row sep=crcr]{%
	16	0.750606079470951\\
	80	0.307208298856149\\
	144	0.0458874443264246\\
	208	0.00321679012894506\\
};
\addlegendentry{TPE RZFBF-ub}
\addplot [mark=triangle, mark size=3pt]
table[row sep=crcr]{%
	16	0.914669288772268\\
	80	0.603698318069811\\
	144	0.342766450077469\\
	208	0.167599320831921\\
	272	0.0610958429556117\\
	336	0.0216359888829325\\
};
\addlegendentry{P ZFBF-lb}

\addplot [dashed, mark=triangle, mark options=solid, color=blue]
table[row sep=crcr]{%
16	0.91487687601062\\
80	0.605816934243825\\
144	0.349305882134221\\
208	0.169859102069559\\
272	0.0676747441437991\\
336	0.0223702533533711\\
};
\addlegendentry{P RZFBF-lb}

\addplot [dashed, mark=diamond, mark options=solid]
table[row sep=crcr]{%
16	0.67243226840392\\
80	0.212477924046338\\
144	0.0263860291051996\\
208	0.00210910595861738\\
};
\addlegendentry{TPE RZFBF-lb}
\end{axis}
\end{scope}
\end{tikzpicture}%